\documentclass[11pt]{article}

\usepackage[final]{acl}
\usepackage{amsmath,amssymb,amsthm,bm,amsfonts}
\usepackage{algorithm, algorithmic}
\usepackage{graphicx,subfig}
\usepackage{multirow}
\usepackage{textcomp}
\usepackage{url,xurl}
\usepackage{color, colortbl}
\usepackage{float, stfloats}
\usepackage{pifont,wasysym}
\usepackage{threeparttable}
\usepackage{soul}
\usepackage{enumitem}
\usepackage{makecell}
\usepackage{verbatim}
\usepackage{xspace}
\usepackage[many]{tcolorbox}
\usepackage{tabularx}
\usepackage{ragged2e}
\usepackage[table,dvipsnames]{xcolor}
\usepackage{enumitem}
\usepackage{svg}
\usepackage{booktabs}
\tcbuselibrary{breakable}
\usepackage[export]{adjustbox}
\usepackage{pgf}
\usepackage{newtxtt}
\usepackage{times}
\usepackage{latexsym}

\usepackage{hyperref}

\newtheorem{lemma}{Lemma}

\newtheorem{proposition}{Proposition}

\usepackage[T1]{fontenc}

\usepackage[utf8]{inputenc}

\usepackage{microtype}

\usepackage{inconsolata}

\usepackage{graphicx}

\newcommand{\aegis}{\textbf{\texttt{AEGIS}}}

\title{\textit{Beyond Over-Refusal}: Defending Indirect Prompt Injection via Latent Instruction Manifolds}

\author{
    Jiahao Chen\textsuperscript{1},
    Rui Yin\textsuperscript{1},
    Xinfeng Li\textsuperscript{2},
    Qianli Ma\textsuperscript{3},
    Tianyu Du\textsuperscript{1} \AND
    Zhihui Fu\textsuperscript{4},
    Jun Wang\textsuperscript{4},
    Zhaoxiang Wang\textsuperscript{4},
    Shouling Ji\textsuperscript{1} \\
    \\
    \textsuperscript{1}Zhejiang University,
    \textsuperscript{2}Hong Kong Polytechnic University,\\
    \textsuperscript{3}University of Science and Technology of China,
    \textsuperscript{4}OPPO Research Institute \\
    \texttt{\{xaddwell,ruiyin,zjradty,sji\}@zju.edu.cn},
    \texttt{qianli-ma@mail.ustc.edu.cn}, \\
    \texttt{\{lxfmakeit,hzzhzzf,junwang.lu,steven.wangzx\}@gmail.com}
}

\begin{document}
\maketitle
\begin{abstract}
Large Language Models (LLMs) have been integrated into complex ecosystems (e.g., Code Agents), while the Indirect Prompt Injection (IPI) attacks have emerged as critical barriers to their safe deployment. Attackers exploit LLMs' indistinguishability between ``instructions'' and ``data'' to manipulate LLMs via maliciously injected instructions. Existing defenses, however, face an intractable safety-utility trade-off: most guardrails either incur high latency or suffer from severe \textit{over-refusal}. In this paper, we first demonstrate that LLMs can separate \textbf{instruction} from \textbf{data} intrinsically with both theoretical and empirical evidence. Inspired by this insight, we propose \aegis{} (\textbf{A}daptive \textbf{E}nsemble \textbf{G}uard for \textbf{I}njection \textbf{S}hielding). \aegis{} extracts instruction-sensitive projectors to identify malicious instructions and leverages a \textbf{Unified Multi-Layer Consensus} mechanism that aggregates topologically distinct signals across the network depth. Empirical evaluations show that \aegis{} achieves remarkable detection performance against both heuristic and optimization-based attacks compared to baselines, highlighting its potential to mitigate IPI. Code is available at \href{https://github.com/xaddwell/AEGIS}{Github}
\end{abstract}

\section{Introduction}
\label{sec:intro}
The development of Large Language Models (LLMs) from passive text generators to active agents has demonstrated remarkable performance across a wide range of Natural Language Processing (NLP) tasks~\cite{taori2023stanford, ouyang2022training}. By integrating with Retrieval-Augmented Generation (RAG) and external tool APIs (e.g., search engines), LLMs now function as core components within complex application ecosystems, capable of executing code and managing user data~\cite{shi2025prompt, zhang2025realistic}. However, this integration exposes a severe vulnerability within LLMs~\cite{willison2022prompt}: the lack of explicit distinction between instructions (control flow) and knowledge (external inputs) within the context.

This ambiguity enables Indirect Prompt Injection (IPI) attacks, where adversaries embed malicious commands into retrieved data streams, such as web pages or emails, to hijack the model's behavior~\cite{willison2022prompt, zhang2025realistic}. Such attacks can lead to severe consequences, including data exfiltration~\cite{lin2026sope}, unauthorized actions, and the propagation of harmful content~\cite{hui2024pleak, collu2025publish}. The demand for robust defense mechanisms has become paramount.

Existing countermeasures, however, struggle to reconcile the \textit{Safety-Utility} trade-off. IPI Sanitization approaches, such as structured queries~\cite{chen2025struq} and token-wise filtering~\cite{geng2025pisanitizer, liu2025drip}, attempt to neutralize threats by restructuring or scrubbing the input stream.  Since these methods necessitate the rigorous filtering of every input, they impose large utility penalties, incurring high computational overheads and often degrading the semantic integrity of complex user queries. Conversely, IPI Detection guardrails aim to identify and block malicious inputs.  However, whether relying on external classifiers~\cite{alon2023detecting, pibert} or internal steering mechanisms~\cite{hung2025attention}, these methods suffer from severe \textit{over-refusal}~\cite{li-etal-2025-piguard} and compromise the model's general availability and helpfulness.

In this work, we challenge the previous assumption that the conflation of instruction and data is indistinguishable for LLM~\cite{zverev2024can}. We argue that the key to resolving this trade-off lies in the intrinsic geometry of the LLM's latent space~\cite{flas2026}, shaped by its two-stage training paradigm. We posit that: Pre-training forces knowledge representations to occupy a high-dimensional, high-entropy manifold to maximize information capacity, whereas \textbf{Instruction Tuning} collapses instruction representations into a separated, compact operator subspace~\cite{park2024linear, zou2023representation}. Consequently, detection should not rely on black-box classification, but on identifying signatures within this compact instruction manifold.

Inspired by these insights, here we introduce \aegis{} (\textbf{A}daptive \textbf{E}nsemble \textbf{G}uard for \textbf{I}njection \textbf{S}hielding), a lightweight detection framework that secures LLMs without training. \aegis{} distinguishes malicious instructions from benign data by filtering out superficial phrasing differences to isolate the core ``imperative intent.'' To further ensure robustness without over-refusal, we introduce a \textit{Unified Multi-Layer Consensus} mechanism. Instead of relying on a single layer, which can be noisy, \aegis{} aggregates evidence across multiple network depths. This ensures that an input is flagged as an attack only if it exhibits a persistent malicious signal throughout the model's processing, eliminating false positives caused by transient noise. Our contributions are summarized as:
\begin{itemize}
    \item We characterize the geometric relationship between instruction and knowledge representations, effectively distinguishing instructions from the high-entropy manifold of the untrusted external knowledge.
    \item We propose \aegis{} that combines distribution-aware threshold calibration with a multi-layer voting mechanism. Beyond empirical efficacy, we derive a \textbf{theoretical error bound} for the False Positive Rate (FPR) and provide a geometric analysis for the ensemble's robustness radius against latent perturbations.
    \item Extensive evaluations on eight IPI attacks with ten baselines demonstrate that \aegis{} achieves remarkable detection performance. Crucially, it improves the safety-utility balance under standard benign workloads, reducing the FPR on benign tasks to negligible levels ($<1\%$).
\end{itemize}

\section{Related Work}
\label{sec:related_works}
\subsection{Instruction Tuning}
Instruction Fine-Tuning is pivotal in aligning pre-trained LLMs with human intent, transitioning them from next-token predictors to helpful assistants~\cite{taori2023stanford, ouyang2022training}. This process typically involves supervising the model on pairs of instructions and desired outputs. While SFT enhances utility and instruction-following capabilities, recent studies suggest it fundamentally alters the geometry of the model's latent space~\cite{zou2023representation}. Previous works have further questioned how well LLMs can intrinsically separate instructions from data~\cite{zverev2024can}. We posit that SFT collapses the representations of imperative instructions into a low-dimensional, compact operator subspace, distinct from the high-entropy manifold of general knowledge acquired during pre-training~\cite{park2024linear}, with more details given in Section~\ref{sec:methodology}.

\subsection{Indirect Prompt Injection Attacks}
IPI attacks exploit the intrinsic inability of LLMs to distinguish between instructions and knowledge~\cite{willison2022prompt, chen-etal-2025-indirect}. These attacks have evolved from manual tricks to complex optimizations targeting various LLM components.

\noindent\textbf{Heuristic-based Attacks.} Early adversarial attempts relied on manual ingenuity and semantic patterns. Attackers employ role-playing scenarios (e.g., ``DAN''), cognitive hacking, or translation wrappers to bypass safety alignment~\cite{perez2022ignore, liu2024promptinjection}. Common techniques include ``Context Ignoring,'' where the model is commanded to disregard prior instructions, and ``Fake Completion,'' which mimics the start of a malicious response to induce non-compliance~\cite{willison2022prompt}. While effective against naive models, these patterns are often brittle and can be mitigated through robust system prompts or supervised safety tuning.

\noindent\textbf{Optimization-based Attacks.} Some adversaries leverage automated optimization to generate adversarial suffixes. Gradient-based methods, such as GCG~\cite{zou2023universal} and activation-guided MCMC sampling~\cite{li-etal-2025-transferable}, search for token sequences that maximize the probability of affirmative responses. Variations like Universal Injection~\cite{liu2024automatic} aim for transferability across models. Beyond text generation, recent works utilize execution-based triggers like NeuralExec~\cite{pasquini2024neural} or target specific applications, such as Tool Selection in agents~\cite{shi2025prompt} and GUI-based agents~\cite{zhang2025realistic}. Additionally, Prompt Leaking strategies (PLeak) aim to extract the system prompt itself~\cite{hui2024pleak}.

\subsection{Indirect Prompt Injection Defenses}
\textbf{IPI Sanitization.} Sanitization neutralizes threats within the untrusted input stream, including structured queries (StruQ) that enforce separation between instructions and data~\cite{chen2025struq}, and PISanitizer, which scrubs malicious patterns~\cite{geng2025pisanitizer}. Advanced methods like DRIP~\cite{liu2025drip} employ de-instruction training and residual fusion, while SecAlign~\cite{chen2025secalign} utilizes preference optimization. Other approaches include task-specific fine-tuning (Jatmo)~\cite{piet2024jatmo} and leveraging defensive tokens~\cite{chen2025defending}. Despite their utility~\cite{an2025ipiguard,du2026snapguard}, sanitization can be computationally expensive and may degrade semantic integrity. Therefore, our paper focuses on the IPI detection for evaluation and comparison.

\noindent\textbf{IPI Detection.} Injection detections classify inputs as benign or malicious before processing~\cite{du2026snapguard}. External guardrails, such as BERT-based classifiers (e.g., PIBert~\cite{pibert}, DeBERTa~\cite{deberta-pi}) or Perplexity-based filters~\cite{alon2023detecting}, operate as separate modules but often lack context-awareness. Internal methods monitor activation patterns; for instance, Attention Tracker~\cite{hung2025attention} and InjecGuard~\cite{li2024injecguard} analyze internal states. Furthermore, recent works have focused on not just detecting but localizing the specific injection segments within the prompt~\cite{jia2026promptlocate, jia2025promptlocate}. However, many existing methods still struggle with the over-refusal problem~\cite{li-etal-2025-piguard}.

\section{Methodology}
\label{sec:methodology}
\subsection{Threat Model}
\label{sec:threat_model}
\noindent\textbf{Adversary's Assumption.}
We assume the adversary has no access to the LLMs' frozen weights $\theta$, the secure system prompt $I_{sys}$, or the legitimate user query $I_{user}$. The adversary's capability is strictly limited to injecting malicious textual instructions ($I_{adv}$) into the external data segment $D_{ext}$ (e.g., by poisoning web pages or emails that are subsequently retrieved by the system)~\cite{willison2022prompt, liu2024promptinjection}. The adversary's objective is to manipulate the model's generation to bypass safety alignment, execute unauthorized commands, or perform data exfiltration~\cite{hui2024pleak}.

\noindent\textbf{Defender's Assumption.}
The defender, as the model owner, possesses white-box access to the LLM during the inference phase. The defender is aware of the input segmentation boundaries (i.e., identifying the external context $D_{ext}$) but has no prior knowledge regarding the benign or malicious nature of that content. To be practical for real-time applications, the defense mechanism must operate under a \textit{Low-Latency} constraint and satisfy a \textit{Utility Preservation} requirement, ensuring that the suppression of malicious instructions does not impair the model's ability to reason over benign knowledge~\cite{li-etal-2025-piguard}.

\begin{figure}
    \centering
    \subfloat[Llama3.1-8B-Base]{\includegraphics[width=0.495\linewidth]{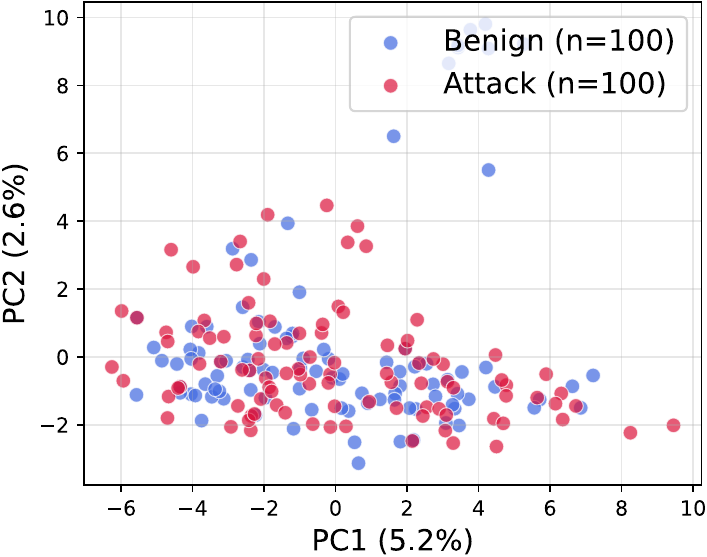}}
    \subfloat[Llama3.1-8B-Instruct]{\includegraphics[width=0.495\linewidth]{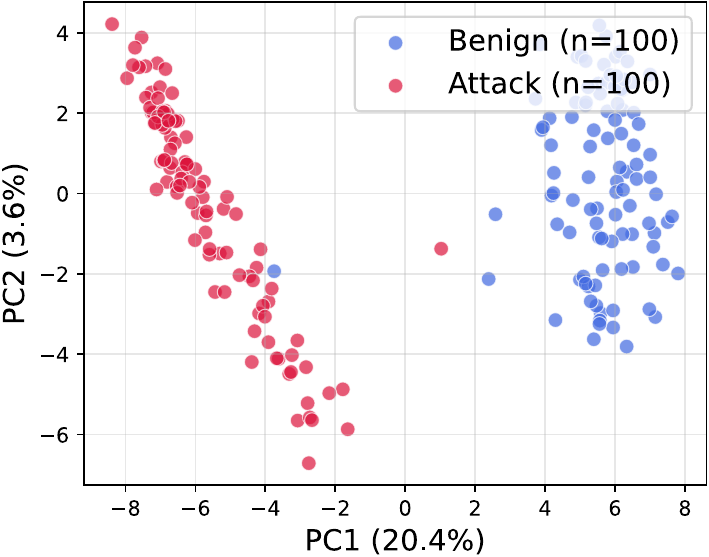}}
    \vspace{-5pt}
    \caption{PCA visualization of the last-token hidden state at layer 25 on two versions of Llama3.1-8B.}
    \label{fig:intuition}
    \vspace{-8pt}
\end{figure}

\subsection{Spectral Asymmetry and SNR}
Pre-training on vast corpora forces the knowledge representations $h_K$ to occupy a high-dimensional, high-variance manifold $\mathcal{M}_K$ to encode diverse facts.
Conversely, Supervised (Instruction) Fine-Tuning (SFT) trains the model to recognize a finite set of operator patterns (e.g., ``summarize'', ``ignore''). Consequently, instruction representations $h_I$ often collapse into a low-dimensional, compact subspace $\mathcal{M}_I$. As shown in Figure~\ref{fig:intuition}, the hidden state of the instruct model (w/ SFT) shows much higher separability than the base one (w/o SFT). Let $\Sigma_K,\Sigma_I\in \mathbb{R}^{d \times d}$ be the covariance matrices of knowledge and instructions. We posit that the spectral norms satisfy $\|\Sigma_K\|_2 \gg \|\Sigma_I\|_2$. Detecting an injection is equivalent to detecting a $h_I$ buried in noise $h_K$. Even if the subspaces overlap, we can maximize the Signal-to-Noise Ratio by finding a projection direction $w$ that minimizes the projected variance of the ``signal class'' (instructions) relative to the distance between class centroids via Fisher's criterion, which motivates our design of \aegis{}.

\subsection{Latent Signal Decomposition}
Consider the activation vector $h^{(l)} \in \mathbb{R}^d$ at layer $l$ for an input sequence $x$. Following the Linear Representation Hypothesis~\cite{park2024linear}, we model $h^{(l)}$ as a superposition of semantic components.
In a RAG context, a benign input consists of a knowledge component $K$. IPIs introduce an adversarial instruction component $I_{adv}$. The latent state can be formulated as:
\begin{equation}
    h^{(l)} = h_K^{(l)} + \alpha \cdot h_{I_{adv}}^{(l)} + \epsilon
\end{equation}
where $\alpha \in \{0, 1\}$ indicates the presence of an attack. The problem of detection is to find a projection function $f: \mathbb{R}^d \to \{0, 1\}$ that maximizes the detection probability $P(f(h)=1 | \alpha=1)$ while bounding the FPR $P(f(h)=1 | \alpha=0) \le \tau$.

\subsection{The \aegis{} Framework}
\subsubsection{Instruction Projector}
To identify IPIs, we must construct a projection operator that isolates the ``instruction intent'' from the background ``knowledge noise.'' 
Let $\mathcal{D}_I$ and $\mathcal{D}_K$ denote the distributions of instruction and knowledge representations, with centroids $\mu_I, \mu_K$ and covariances $\Sigma_I, \Sigma_K$. A naïve separation (e.g., difference of means) is insufficient because instruction representations exhibit significant variance along specific ``synonymy axes'' (e.g., lexical variations like ``Summarize'' vs. ``TL;DR'') while remaining invariant along the ``intent axis.''
We aim to find a projection direction $w$ that maximizes the separation between the instruction signal and the knowledge, while simultaneously minimizing the variance \textit{within} the instruction class. This ensures the detector locks onto the stable semantic core of an imperative command rather than superficial phrasing. We formulate this as:
\begin{equation}
    w^* = \mathop{\arg\max}_{w} \frac{(w^T (\mu_I - \mu_K))^2}{w^T \Sigma_I w},
    \label{eq:insm}
\end{equation}
where we use Ledoit-Wolf shrinkage to estimate $\Sigma_I, \Sigma_K$. Here, the numerator rewards the separation of semantic centers, while the denominator penalizes directions where instructions are unstable or noisy. 
Solving this optimization via Lagrange multipliers yields the closed-form \textbf{Instruction-Sensitive Projector}:
\begin{equation}
    w^* \propto \Sigma_I^{-1} (\mu_I - \mu_K)
\end{equation}
The term $\Sigma_I^{-1}$ acts as a \textbf{Semantic Focuser} and down-weights the dimensions where instructions vary (trivial lexical features) and up-weights the dimensions where instructions are compact and distinct from knowledge. Note that we discard $\Sigma_K$ in the optimization to avoid the futile compression of the knowledge manifold, focusing on collapsing the instruction manifold into a compact `spike'.

\subsubsection{Distribution-Aware Safety Calibration}
Having extracted the optimal detection direction $w^*$, we define a decision boundary. A fixed geometric margin is hazardous due to the heavy-tailed nature of the knowledge distribution.
To avoid over-refusal without making parametric assumptions about the knowledge manifold, we introduce a \textbf{Distribution-Aware Thresholding} mechanism.
We define the decision boundary $b$ based on the observed projection statistics of the benign knowledge set $\mathcal{D}_{benign}$. For a required safety level $\beta$ (e.g., guaranteeing $99\%$ utility retention), we set:
\begin{equation}
    b = - \mathcal{Q}_{1-\beta}(\{ w^{*T} h \mid h \in \mathcal{D}_{benign} \})
\end{equation}
where $\mathcal{Q}_{p}(\cdot)$ denotes the $p$-th quantile function. This mechanism ensures that the detector's sensitivity is mathematically anchored to the model's actual behavior on safe inputs so that the detector will trigger on at most $\beta$ fraction of benign queries.

\subsubsection{Multi-Layer Consensus Voting}
Single-layer detectors, while calibrated, may exhibit localized sensitivity to noise or ``polysemantic'' activation overlaps. To capture this global consistency, we model the layers as a committee and propose a \textbf{Unified Consensus Mechanism}. Let $S_l(x) = w_l^{*T} h^{(l)}$ be the raw projection score at layer $l \in \mathcal{L}_{vote}$. We define the ensemble score $\mathcal{S}_{ens}$ as the aggregate of transformed layer-wise signals:
\begin{equation}
    \mathcal{S}_{ens}(x) = \sum_{l \in \mathcal{L}_{vote}} \phi_l( S_l(x) )
\end{equation}
where $\phi_l(\cdot)$ is a layer-specific transformation kernel. The final detection decision is given by $\hat{y} = \mathbb{I}(\mathcal{S}_{ens}(x) \ge \Gamma)$, where $\Gamma$ is the global consensus threshold.
Our framework supports two distinct instantiation modes for $\phi_l$: (1) $\phi_{hard}(s) = \mathbb{I}(s > b_l)$
where $b_l$ is the FPR-calibrated threshold and $\mathcal{S}_{ens}$ represents the raw vote count, and $\Gamma$ can be the consensus quorum; (2) $\phi_{soft}(s) = (s - \mu_{benign}^{(l)})/\sigma_{b}^{(l)}$
where $\mu_{b}^{(l)}$ and $\sigma_{b}^{(l)}$ are the mean and standard deviation of benign projection scores at layer $l$. The \textbf{visual conceptual illustration} of \aegis{} framework is given in Figure~\ref{fig:framework} (for space limit).

\subsection{Theoretical Error Bound of \aegis{}}
\label{subsec:theoretical_analysis}
Consider the \textbf{Hard Voting} mode, we define the binary indicator variable $I_l = \phi_{hard}(S_l(x))$. For a benign input, we have the bound $P(I_l=1) \le \beta$.
Assuming approximate independence~\cite{abs-2502-02790} between non-adjacent layers, the ensemble FPR, denoted as $Q_{ens} = P(\mathcal{S}_{ens} \ge \Gamma)$, follows the tail of a Binomial distribution. By applying \textbf{Hoeffding's Inequality}~\cite{hoeffding1963probability}, we derive an exponential upper bound for the system-wide error rate:
\begin{equation}
    Q_{ens} \le \exp\left( -2 |\mathcal{L}_{vote}| \left( \frac{\Gamma}{|\mathcal{L}_{vote}|} - \beta \right)^2 \right)
    \label{eq:error_bound}
\end{equation}
This bound implies that as long as the consensus quorum ratio $\Gamma/|\mathcal{L}_{vote}|$ (e.g., 0.5 for majority rule) exceeds the single-layer error rate $\beta$ (e.g., 0.1), the probability of a false alarm decays exponentially with the number of layers. This guarantee allows us to set a looser $\beta$ to maintain high Recall at the layer level without compromising global safety. The analysis of soft voting is given in Appendix~\ref{app:stable_soft}.

\subsection{Analysis of Latent Ensemble Robustness}
\label{subsec:theoretical_robustness}
Here, we derive a lower bound on the perturbation magnitude required to deceive the ensemble in \textbf{Hard Voting} mode. We consider a \textit{Latent Adversary} capable of introducing an additive latent perturbation vector $\delta \in \mathbb{R}^d$ directly to the hidden state $h^{(l)}(x)$. The goal of the adversary is to force a safe input $x$ to cross the detection boundary defined by the projector $w_l^*$ and threshold $b_l$. Recall that our scoring function at layer $l$: $S_l(h) = \langle h, w_l^* \rangle$, and a detection is triggered if $S_l(h) \ge b_l$. Since we enforce $\|w_l^*\|_2 = 1$, the geometry of the decision boundary is a hyperplane with $w_l^*$.

For a benign input $h^{(l)}(x)$ correctly classified as safe (i.e., $\langle h^{(l)}(x), w_l^* \rangle < b_l$), the \textbf{Latent Robustness Radius}, denoted as $R_l(x)$, is defined as the minimum Euclidean distance from the current representation to the decision hyperplane:
\begin{equation}
    \begin{aligned}
    R_l(x) &= \min_{\delta} \|\delta\|_2 \quad \text{s.t.} \quad \langle h^{(l)}(x) + \delta, w_l^* \rangle \ge b_l \\
    &= \frac{b_l - \langle h^{(l)}(x), w_l^* \rangle}{\|w_l^*\|_2} = b_l - S_l(x)
    \end{aligned}
    \label{eq:exact_radius}
\end{equation}
This equation connects our optimization to robustness: the numerator $b_l - S_l(x)$ is the margin that Eq.~\ref{eq:insm} aims to maximize. By pushing the instruction manifold away from the benign knowledge distribution, \aegis{} enlarges this radius.

\subsubsection{Theoretical Robustness of the Ensemble}
We now extend to the multi-layer ensemble. An adversary must successfully flip the decision of at least $\Gamma$ layers to trigger a system-wide false positive, a simultaneous constraint on the perturbation.

\begin{proposition}[Ensemble Latent Robustness Radius]
    Let $\mathcal{R}(x) = \{R_1(x), R_2(x), \dots, R_{|\mathcal{L}|}(x)\}$ be the set of exact latent radii for all monitored layers calculated via Eq.~(\ref{eq:exact_radius}). Let $R_{(\Gamma)}$ denote the $\Gamma$-th smallest value in $\mathcal{R}(x)$ (i.e., the $\Gamma$-th order statistic).
    The latent robustness radius of the \aegis{} ensemble, $R_{ens}(x)$, is exactly determined by:
    \begin{equation}
        R_{ens}(x) = R_{(\Gamma)}
    \end{equation}
    Any latent perturbation $\delta$ with $\|\delta\|_2 < R_{(\Gamma)}$ is mathematically guaranteed to preserve the safe classification of the ensemble.
\end{proposition}

\textbf{Geometric Interpretation.}
The proposition above formalizes that the ensemble is as strong as its ``median'' link, rather than its weakest. Even if a specific layer $k$ has a small margin $R_k \approx 0$, the attacker must simultaneously satisfy the geometric constraints of $\Gamma$ in different semantic spaces. Finding a single vector $\delta$ that is collinear with $\Gamma$, different normals is geometrically constrained, resulting in a larger effective radius $R_{(\Gamma)}$.

\section{Experiments}
\subsection{Experimental Settings}
\noindent\textbf{Evaluation Datasets.}
Following previous works~\cite{jia2025promptlocate,hung2025attention}, we utilize the \textit{OpenPromptInjection} benchmark~\cite{liu2024promptinjection} for evaluation with 7 NLP tasks: duplicate sentence detection, grammar correction, hate speech detection, natural language inference, sentiment analysis, spam detection, and text summarization, with 100 instruction-data pairs per task.

\noindent\textbf{IPIs.} We evaluate \aegis{} against a diverse spectrum of IPIs:
(1) \textit{Heuristic-based attacks}: following the benchmark~\cite{liu2024promptinjection}, we include Naive Attack~\cite{harang2023securing}, Escape Characters~\cite{perez2022ignore}, Context Ignoring~\cite{willison2022prompt}, Fake Completion~\cite{willison2022prompt}, and Combined Attack~\cite{liu2024promptinjection}.
(2) \textit{Optimization-based attacks}: we employ Universal Injection~\cite{liu2024automatic}, NeuralExec~\cite{pasquini2024neural}, and PLeak~\cite{hui2024pleak}. All attacks utilize standard parameter configurations from their respective open-source implementations. Each attack was implemented on the datasets provided above (injected with the dataset samples above), consisting of $7\times 7\times 100\times 8$ attack samples.

\noindent\textbf{Evaluation Metrics.} Following previous works~\cite{li-etal-2025-piguard,liu2024promptinjection}, we use accuracy (ACC), false positive rate (FPR), false negative rate (FNR), true positive rate (TPR), F1 score, and computation overhead (time, GPU memory) as metrics. Note that FNR measures the fraction of the malicious samples that are falsely detected as clean.

\noindent\textbf{Defense Implementations.} We leverage the a balanced training corpus~\cite{zou2025pishield}, consisting of 10,000 benign samples from the processed Alpaca~\cite{taori2023stanford} and Natural Questions~\cite{kwiatkowski2019natural} with 10,000 negative sampels from the Naive Attack~\cite{liu2024promptinjection} and NeuralExec~\cite{pasquini2024neural}. Specifically, we only select 200 samples randomly from the corpus for $w^*$ estimation. For evaluated models, we include three mainstream models including Qwen3-4B (default in evaluation unless specified)~\cite{qwen3technicalreport}, Llama3-8B~\cite{llama3modelcard} and Llama3.1-8B~\cite{grattafiori2024llama}. More details are given in Appendix~\ref{app:aegis_detail}.

\noindent\textbf{Baselines.} We compare our method with both industrial and research detections, including (1) PIBert~\cite{pibert}, ProtectAI-deberta (DeBERTa)~\cite{deberta-pi}, and PromptGuard (PrGuard)~\cite{PromptGuard}; (2) PPL~\cite{alon2023detecting} (GPT2 for PPL$^1$ and Vicuna-7B for PPL$^2$), Known-Answer (KA)~\cite{Known-Answer}, InjecGuard~\cite{li-etal-2025-piguard,li2024injecguard}, AttentionTracker (AttTracker)~\cite{hung2025attention}, LLM-Naive~\cite{llm-naive} and PIShield~\cite{zou2025pishield}. The details and implementation of these baselines can be found in Appendix~\ref{app:baseline_detail}.

\begin{table*}[t]
\centering
\setlength{\tabcolsep}{0.8mm}
\renewcommand{\arraystretch}{1.05}
\aboverulesep=0.1ex
\belowrulesep=0.3ex
\resizebox{1.0\linewidth}{!}{
\begin{tabular}{ccccccccccccc}
\toprule[1.5pt]
\textbf{Attack} $\downarrow$ & \textbf{Defense}$\rightarrow$ & \textbf{AttTracker} & \textbf{DeBERTa} & \textbf{InjecGuard} & \textbf{KA} & \textbf{LLM Naive} & \textbf{PIBert} & \textbf{PPL$^{1}$} & \textbf{PPL$^{2}$} & \textbf{PrGuard} & \textbf{PIShield} & \textbf{AEGIS} \\ 
\midrule
 & \cellcolor[HTML]{EFEFEF}ACC $\uparrow$ & \cellcolor[HTML]{EFEFEF}{ 0.8071} & \cellcolor[HTML]{EFEFEF}{ 0.6219} & \cellcolor[HTML]{EFEFEF}{ \ul{0.9836}} & \cellcolor[HTML]{EFEFEF}{ 0.7195} & \cellcolor[HTML]{EFEFEF}{ 0.6710} & \cellcolor[HTML]{EFEFEF}{ 0.6050} & \cellcolor[HTML]{EFEFEF}{ 0.4508} & \cellcolor[HTML]{EFEFEF}{ 0.5182} & \cellcolor[HTML]{EFEFEF}{ 0.5371} & \cellcolor[HTML]{EFEFEF}{0.5758} & \cellcolor[HTML]{EFEFEF}{ {\textbf{0.9848}}} \\
 & \cellcolor[HTML]{ECF4FF}FPR $\downarrow$ & \cellcolor[HTML]{ECF4FF}0.3857 & \cellcolor[HTML]{ECF4FF}{\ul{0.0057}} & \cellcolor[HTML]{ECF4FF}\textbf{0.0243} & \cellcolor[HTML]{ECF4FF}0.4931 & \cellcolor[HTML]{ECF4FF}0.6069 & \cellcolor[HTML]{ECF4FF}0.7900 & \cellcolor[HTML]{ECF4FF}0.5492 & \cellcolor[HTML]{ECF4FF}0.7100 & \cellcolor[HTML]{ECF4FF}0.9257 & \cellcolor[HTML]{ECF4FF}{0.8171} & \cellcolor[HTML]{ECF4FF}0.0290 \\
\multirow{-3}{*}{\makecell{\textbf{Combined}\\\textbf{Attack}}} & \cellcolor[HTML]{FFFFC7}FNR $\downarrow$ & \cellcolor[HTML]{FFFFC7}\textbf{0.0000} & \cellcolor[HTML]{FFFFC7}0.7504 & \cellcolor[HTML]{FFFFC7}0.0086 & \cellcolor[HTML]{FFFFC7}0.0680 & \cellcolor[HTML]{FFFFC7}0.0510 & \cellcolor[HTML]{FFFFC7}\textbf{0.0000} & \cellcolor[HTML]{FFFFC7}0.5492 & \cellcolor[HTML]{FFFFC7}0.2537 & \cellcolor[HTML]{FFFFC7}\textbf{0.0000} & \cellcolor[HTML]{FFFFC7}{0.0312} & \cellcolor[HTML]{FFFFC7}{\ul{0.0014}} \\ \midrule
 & \cellcolor[HTML]{EFEFEF}ACC $\uparrow$ & \cellcolor[HTML]{EFEFEF}{\ul{0.8071}} & \cellcolor[HTML]{EFEFEF}0.4996 & \cellcolor[HTML]{EFEFEF}0.6005 & \cellcolor[HTML]{EFEFEF}0.6701 & \cellcolor[HTML]{EFEFEF}0.5560 & \cellcolor[HTML]{EFEFEF}0.6050 & \cellcolor[HTML]{EFEFEF}0.3996 & \cellcolor[HTML]{EFEFEF}0.4597 & \cellcolor[HTML]{EFEFEF}0.5371 & \cellcolor[HTML]{EFEFEF}{0.5790} & \cellcolor[HTML]{EFEFEF}\textbf{0.9852} \\
 & \cellcolor[HTML]{ECF4FF}FPR $\downarrow$ & \cellcolor[HTML]{ECF4FF}0.3857 & \cellcolor[HTML]{ECF4FF}{\ul{0.0057}} & \cellcolor[HTML]{ECF4FF}\textbf{0.0243} & \cellcolor[HTML]{ECF4FF}0.4967 & \cellcolor[HTML]{ECF4FF}0.6029 & \cellcolor[HTML]{ECF4FF}0.7900 & \cellcolor[HTML]{ECF4FF}0.6004 & \cellcolor[HTML]{ECF4FF}0.7100 & \cellcolor[HTML]{ECF4FF}0.9257 & \cellcolor[HTML]{ECF4FF}{0.8171} & \cellcolor[HTML]{ECF4FF}0.0290 \\
\multirow{-3}{*}{\makecell{\textbf{Escape}\\\textbf{Character}}} & \cellcolor[HTML]{FFFFC7}FNR $\downarrow$ & \cellcolor[HTML]{FFFFC7}\textbf{0.0000} & \cellcolor[HTML]{FFFFC7}0.9951 & \cellcolor[HTML]{FFFFC7}0.7747 & \cellcolor[HTML]{FFFFC7}0.1631 & \cellcolor[HTML]{FFFFC7}0.2851 & \cellcolor[HTML]{FFFFC7}\textbf{0.0000} & \cellcolor[HTML]{FFFFC7}0.6004 & \cellcolor[HTML]{FFFFC7}0.3706 & \cellcolor[HTML]{FFFFC7}\textbf{0.0000} & \cellcolor[HTML]{FFFFC7}{0.0249} & \cellcolor[HTML]{FFFFC7}{\ul{0.0006}} \\ 
\midrule
& \cellcolor[HTML]{EFEFEF}ACC $\uparrow$ & \cellcolor[HTML]{EFEFEF}{\ul{0.8071}} & \cellcolor[HTML]{EFEFEF}0.4999 & \cellcolor[HTML]{EFEFEF}0.6124 & \cellcolor[HTML]{EFEFEF}0.7082 & \cellcolor[HTML]{EFEFEF}0.6254 & \cellcolor[HTML]{EFEFEF}0.6050 & \cellcolor[HTML]{EFEFEF}0.4347 & \cellcolor[HTML]{EFEFEF}0.5051 & \cellcolor[HTML]{EFEFEF}0.5371 & \cellcolor[HTML]{EFEFEF}{0.5769} & \cellcolor[HTML]{EFEFEF}\textbf{0.9848} \\
 & \cellcolor[HTML]{ECF4FF}FPR $\downarrow$ & \cellcolor[HTML]{ECF4FF}0.3857 & \cellcolor[HTML]{ECF4FF}{\ul{0.0057}} & \cellcolor[HTML]{ECF4FF}\textbf{0.0243} & \cellcolor[HTML]{ECF4FF}0.4963 & \cellcolor[HTML]{ECF4FF}0.6024 & \cellcolor[HTML]{ECF4FF}0.7900 & \cellcolor[HTML]{ECF4FF}0.5653 & \cellcolor[HTML]{ECF4FF}0.7100 & \cellcolor[HTML]{ECF4FF}0.9257 & \cellcolor[HTML]{ECF4FF}{0.8171} & \cellcolor[HTML]{ECF4FF}0.0290 \\
\multirow{-3}{*}{\makecell{\textbf{Fake}\\\textbf{Completion}}} & \cellcolor[HTML]{FFFFC7}FNR $\downarrow$ & \cellcolor[HTML]{FFFFC7}\textbf{0.0000} & \cellcolor[HTML]{FFFFC7}0.9945 & \cellcolor[HTML]{FFFFC7}0.7508 & \cellcolor[HTML]{FFFFC7}0.0873 & \cellcolor[HTML]{FFFFC7}0.1467 & \cellcolor[HTML]{FFFFC7}\textbf{0.0000} & \cellcolor[HTML]{FFFFC7}0.5653 & \cellcolor[HTML]{FFFFC7}0.2798 & \cellcolor[HTML]{FFFFC7}\textbf{0.0000} & \cellcolor[HTML]{FFFFC7}{0.0290} & \cellcolor[HTML]{FFFFC7}{\ul{0.0013}} \\ \midrule
 & \cellcolor[HTML]{EFEFEF}ACC $\uparrow$ & \cellcolor[HTML]{EFEFEF}0.8071 & \cellcolor[HTML]{EFEFEF}0.6873 & \cellcolor[HTML]{EFEFEF}\ul{0.9823} & \cellcolor[HTML]{EFEFEF}0.7056 & \cellcolor[HTML]{EFEFEF}0.4616 & \cellcolor[HTML]{EFEFEF}0.6050 & \cellcolor[HTML]{EFEFEF}0.4535 & \cellcolor[HTML]{EFEFEF}0.4990 & \cellcolor[HTML]{EFEFEF}0.5371 & \cellcolor[HTML]{EFEFEF}{0.5755} & \cellcolor[HTML]{EFEFEF}{\textbf{0.9850}} \\
 & \cellcolor[HTML]{ECF4FF}FPR $\downarrow$ & \cellcolor[HTML]{ECF4FF}0.3857 & \cellcolor[HTML]{ECF4FF}{\ul{0.0057}} & \cellcolor[HTML]{ECF4FF}\textbf{0.0243} & \cellcolor[HTML]{ECF4FF}0.4971 & \cellcolor[HTML]{ECF4FF}0.6155 & \cellcolor[HTML]{ECF4FF}0.7900 & \cellcolor[HTML]{ECF4FF}0.5465 & \cellcolor[HTML]{ECF4FF}0.7100 & \cellcolor[HTML]{ECF4FF}0.9257 & \cellcolor[HTML]{ECF4FF}{0.8171} & \cellcolor[HTML]{ECF4FF}0.0290 \\
\multirow{-3}{*}{\makecell{\textbf{Context}\\\textbf{Ignoring}}} & \cellcolor[HTML]{FFFFC7}FNR $\downarrow$ & \cellcolor[HTML]{FFFFC7}\textbf{0.0000} & \cellcolor[HTML]{FFFFC7}0.6196 & \cellcolor[HTML]{FFFFC7}0.0110 & \cellcolor[HTML]{FFFFC7}0.0916 & \cellcolor[HTML]{FFFFC7}0.4612 & \cellcolor[HTML]{FFFFC7}\textbf{0.0000} & \cellcolor[HTML]{FFFFC7}0.5465 & \cellcolor[HTML]{FFFFC7}0.2920 & \cellcolor[HTML]{FFFFC7}\textbf{0.0000} & \cellcolor[HTML]{FFFFC7}{0.0318} & \cellcolor[HTML]{FFFFC7}{\ul{0.0021}} \\ \midrule
 & \cellcolor[HTML]{EFEFEF}ACC $\uparrow$ & \cellcolor[HTML]{EFEFEF}0.8071 & \cellcolor[HTML]{EFEFEF}0.4996 & \cellcolor[HTML]{EFEFEF}0.6005 & \cellcolor[HTML]{EFEFEF}{\ul{0.7018}} & \cellcolor[HTML]{EFEFEF}0.4585 & \cellcolor[HTML]{EFEFEF}0.6050 & \cellcolor[HTML]{EFEFEF}0.4231 & \cellcolor[HTML]{EFEFEF}0.4752 & \cellcolor[HTML]{EFEFEF}0.5371 & \cellcolor[HTML]{EFEFEF}{0.5783} & \cellcolor[HTML]{EFEFEF}\textbf{0.9850} \\
 & \cellcolor[HTML]{ECF4FF}FPR $\downarrow$ & \cellcolor[HTML]{ECF4FF}0.3857 & \cellcolor[HTML]{ECF4FF}{\ul{0.0057}} & \cellcolor[HTML]{ECF4FF}\textbf{0.0243} & \cellcolor[HTML]{ECF4FF}0.4994 & \cellcolor[HTML]{ECF4FF}0.6051 & \cellcolor[HTML]{ECF4FF}0.7900 & \cellcolor[HTML]{ECF4FF}0.5769 & \cellcolor[HTML]{ECF4FF}0.7100 & \cellcolor[HTML]{ECF4FF}0.9257 & \cellcolor[HTML]{ECF4FF}{0.8171} & \cellcolor[HTML]{ECF4FF}0.0290 \\
\multirow{-3}{*}{\makecell{\textbf{Naive}\\\textbf{Attack}}} & \cellcolor[HTML]{FFFFC7}FNR $\downarrow$ & \cellcolor[HTML]{FFFFC7}\textbf{0.0000} & \cellcolor[HTML]{FFFFC7}0.9951 & \cellcolor[HTML]{FFFFC7}0.7747 & \cellcolor[HTML]{FFFFC7}0.0969 & \cellcolor[HTML]{FFFFC7}0.4780 & \cellcolor[HTML]{FFFFC7}\textbf{0.0000} & \cellcolor[HTML]{FFFFC7}0.5769 & \cellcolor[HTML]{FFFFC7}0.3396 & \cellcolor[HTML]{FFFFC7}\textbf{0.0000} & \cellcolor[HTML]{FFFFC7}{0.0263} & \cellcolor[HTML]{FFFFC7}{\ul{0.0010}} \\ \midrule
 & \cellcolor[HTML]{EFEFEF}ACC $\uparrow$ & \cellcolor[HTML]{EFEFEF}0.8071 & \cellcolor[HTML]{EFEFEF}0.5620 & \cellcolor[HTML]{EFEFEF}\textbf{0.9879} & \cellcolor[HTML]{EFEFEF}0.7404 & \cellcolor[HTML]{EFEFEF}0.6964 & \cellcolor[HTML]{EFEFEF}0.6050 & \cellcolor[HTML]{EFEFEF}0.7867 & \cellcolor[HTML]{EFEFEF}0.6413 & \cellcolor[HTML]{EFEFEF}0.5371 & \cellcolor[HTML]{EFEFEF}{0.5786} & \cellcolor[HTML]{EFEFEF}{\ul{0.9847}} \\
 & \cellcolor[HTML]{ECF4FF}FPR $\downarrow$ & \cellcolor[HTML]{ECF4FF}0.3857 & \cellcolor[HTML]{ECF4FF}{\ul{0.0057}} & \cellcolor[HTML]{ECF4FF}\textbf{0.0243} & \cellcolor[HTML]{ECF4FF}0.4957 & \cellcolor[HTML]{ECF4FF}0.6020 & \cellcolor[HTML]{ECF4FF}0.7900 & \cellcolor[HTML]{ECF4FF}0.2133 & \cellcolor[HTML]{ECF4FF}0.7100 & \cellcolor[HTML]{ECF4FF}0.9257 & \cellcolor[HTML]{ECF4FF}{0.8171} & \cellcolor[HTML]{ECF4FF}\textbf{0.029} \\
\multirow{-3}{*}{\textbf{NeuralExec}} & \cellcolor[HTML]{FFFFC7}FNR $\downarrow$ & \cellcolor[HTML]{FFFFC7}\textbf{0.0000} & \cellcolor[HTML]{FFFFC7}0.8702 & \cellcolor[HTML]{FFFFC7}\textbf{0.0000} & \cellcolor[HTML]{FFFFC7}0.0235 & \cellcolor[HTML]{FFFFC7}0.0051 & \cellcolor[HTML]{FFFFC7}\textbf{0.0000} & \cellcolor[HTML]{FFFFC7}0.2133 & \cellcolor[HTML]{FFFFC7}0.0073 & \cellcolor[HTML]{FFFFC7}\textbf{0.0000} & \cellcolor[HTML]{FFFFC7}{0.0257} & \cellcolor[HTML]{FFFFC7}{\ul{0.0016}} \\ \midrule
 & \cellcolor[HTML]{EFEFEF}ACC $\uparrow$ & \cellcolor[HTML]{EFEFEF}0.8071 & \cellcolor[HTML]{EFEFEF}0.8170 & \cellcolor[HTML]{EFEFEF}{\ul{0.8429}} & \cellcolor[HTML]{EFEFEF}0.6180 & \cellcolor[HTML]{EFEFEF}0.6391 & \cellcolor[HTML]{EFEFEF}0.6050 & \cellcolor[HTML]{EFEFEF}0.6271 & \cellcolor[HTML]{EFEFEF}0.6106 & \cellcolor[HTML]{EFEFEF}0.5371 & \cellcolor[HTML]{EFEFEF}{0.5893} & \cellcolor[HTML]{EFEFEF}\textbf{0.9855} \\
 & \cellcolor[HTML]{ECF4FF}FPR $\downarrow$ & \cellcolor[HTML]{ECF4FF}0.3857 & \cellcolor[HTML]{ECF4FF}{\ul{0.0057}} & \cellcolor[HTML]{ECF4FF}\textbf{0.0243} & \cellcolor[HTML]{ECF4FF}0.4945 & \cellcolor[HTML]{ECF4FF}0.6035 & \cellcolor[HTML]{ECF4FF}0.7900 & \cellcolor[HTML]{ECF4FF}0.3729 & \cellcolor[HTML]{ECF4FF}0.7100 & \cellcolor[HTML]{ECF4FF}0.9257 & \cellcolor[HTML]{ECF4FF}{0.8171} & \cellcolor[HTML]{ECF4FF}0.0290 \\
\multirow{-3}{*}{\textbf{PLeak}} & \cellcolor[HTML]{FFFFC7}FNR $\downarrow$ & \cellcolor[HTML]{FFFFC7}\textbf{0.0000} & \cellcolor[HTML]{FFFFC7}0.3602 & \cellcolor[HTML]{FFFFC7}0.2900 & \cellcolor[HTML]{FFFFC7}0.2696 & \cellcolor[HTML]{FFFFC7}0.1184 & \cellcolor[HTML]{FFFFC7}\textbf{0.0000} & \cellcolor[HTML]{FFFFC7}0.3729 & \cellcolor[HTML]{FFFFC7}0.0688 & \cellcolor[HTML]{FFFFC7}\textbf{0.0000} & \cellcolor[HTML]{FFFFC7}{0.0043} & \cellcolor[HTML]{FFFFC7}\textbf{0.0000} \\ \midrule
 & \cellcolor[HTML]{EFEFEF}ACC $\uparrow$ & \cellcolor[HTML]{EFEFEF}0.8071 & \cellcolor[HTML]{EFEFEF}0.5304 & \cellcolor[HTML]{EFEFEF}\textbf{0.9879} & \cellcolor[HTML]{EFEFEF}0.7371 & \cellcolor[HTML]{EFEFEF}0.6118 & \cellcolor[HTML]{EFEFEF}0.6050 & \cellcolor[HTML]{EFEFEF}0.5963 & \cellcolor[HTML]{EFEFEF}0.6029 & \cellcolor[HTML]{EFEFEF}0.5371 & \cellcolor[HTML]{EFEFEF}{0.5774} & \cellcolor[HTML]{EFEFEF}{\ul{0.9850}} \\
 & \cellcolor[HTML]{ECF4FF}FPR $\downarrow$ & \cellcolor[HTML]{ECF4FF}0.3857 & \cellcolor[HTML]{ECF4FF}{\ul{0.0057}} & \cellcolor[HTML]{ECF4FF}\textbf{0.0243} & \cellcolor[HTML]{ECF4FF}0.4922 & \cellcolor[HTML]{ECF4FF}0.6035 & \cellcolor[HTML]{ECF4FF}0.7900 & \cellcolor[HTML]{ECF4FF}0.4037 & \cellcolor[HTML]{ECF4FF}0.7100 & \cellcolor[HTML]{ECF4FF}0.9257 & \cellcolor[HTML]{ECF4FF}{0.8171} & \cellcolor[HTML]{ECF4FF}0.0290 \\
\multirow{-3}{*}{\makecell{\textbf{Universal}\\\textbf{Injection}}} & \cellcolor[HTML]{FFFFC7}FNR $\downarrow$ & \cellcolor[HTML]{FFFFC7}\textbf{0.0000} & \cellcolor[HTML]{FFFFC7}0.9335 & \cellcolor[HTML]{FFFFC7}\textbf{0.0000} & \cellcolor[HTML]{FFFFC7}0.0335 & \cellcolor[HTML]{FFFFC7}0.1729 & \cellcolor[HTML]{FFFFC7}\textbf{0.0000} & \cellcolor[HTML]{FFFFC7}0.4037 & \cellcolor[HTML]{FFFFC7}0.0843 & \cellcolor[HTML]{FFFFC7}\textbf{0.0000} & \cellcolor[HTML]{FFFFC7}{0.0280} & \cellcolor[HTML]{FFFFC7}{\ul{0.0010}} \\ 
\bottomrule[1pt]
\end{tabular}}
\caption{Overall detection performance (ACC, FPR and FNR) comparison between AEGIS and baselines.}
\vspace{-8pt}
\label{tab:main_opi}
\end{table*}

\subsection{Main Results}
\textbf{Performance Comparison.} Table~\ref{tab:main_opi} presents the comprehensive evaluation results across 8 IPIs. We observe that \aegis{} consistently outperforms baseline methods, achieving SOTA ACC while maintaining a negligible FPR. Note that FPR of all detection against different attacks is almost the same since the benign samples to calculate FPR are the same for different attacks. The primary limitation of existing defenses is the imbalance between detection sensitivity and benign utility preservation. As shown in Table~\ref{tab:main_opi}, baseline methods typically gravitate towards one extreme: (1) \textit{Over-Refusal:} Methods like AttTracker and PrGuard achieve low FNR but suffer from high FPRs (38.57\% and 92.57\%, respectively), indicating a large portion of benign queries are incorrectly rejected. (2) \textit{Under-Detection:} Conversely, shallow classifiers such as DeBERTa and PIBert maintain low FPRs ($\approx 0.5\%$) but fail to generalize to heuristic attacks, exhibiting FNRs as high as 99.51\% on Naive Attacks and 99.45\% on Fake Completion. In contrast, \aegis{} achieves an average ACC exceeding 98\% across most tasks. Crucially, it maintains a consistent FPR of \textbf{0.29\%}. This confirms that \aegis{} can identify malicious intent without compromising the model's ability to process safe queries. The score distribution of \aegis{} on different datasets and attacks be found in Figure~\ref{fig:hist_distribution}.

Also, \aegis{} demonstrates superior generalization capabilities against both attack paradigms. It effectively detects \textit{heuristic-based attacks}, where static embeddings often fail due to lexical variations. Furthermore, it remains robust against sophisticated \textit{optimization-based attacks}, achieving $>98\%$ accuracy where perplexity-based methods struggle (e.g., PPL$^1$ achieves only 59.63\% accuracy on Universal Injection). This validates the effectiveness of the multi-layer consensus mechanism in capturing deep semantic anomalies that persist across the network depth.

\begin{figure}[t]
    \centering
    \includegraphics[width=1\linewidth]{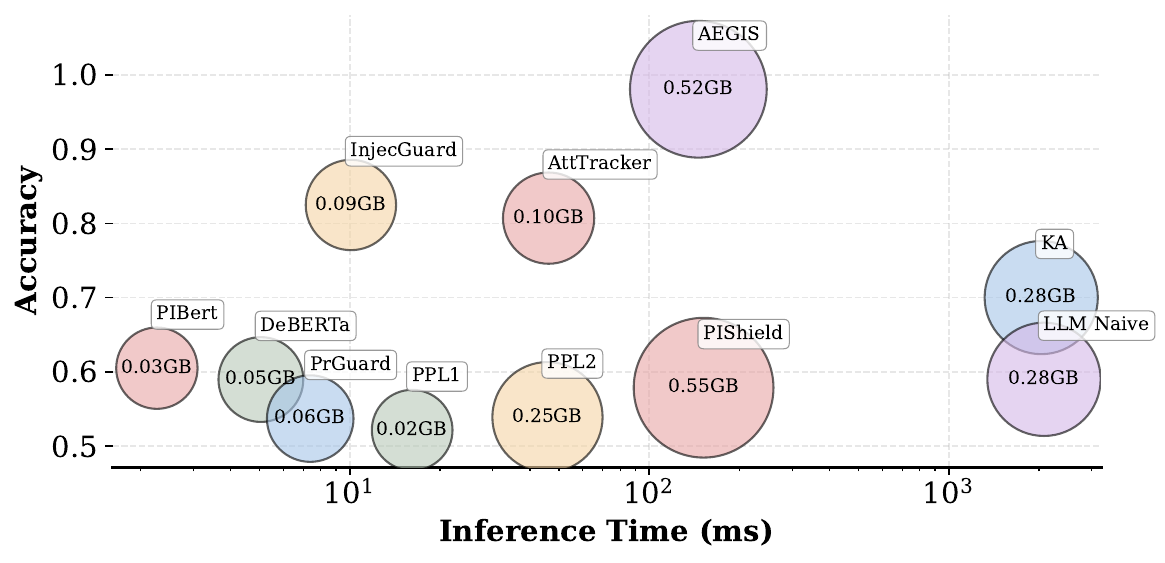}
    \vspace{-15pt}
    \caption{Efficiency-accuracy trade-off comparison.}
    \label{fig:efficiency}
    \vspace{-10pt}
\end{figure}

\textbf{Overhead Comparison.} Results from Qwen3-4B in Figure~\ref{fig:efficiency} visualizes the trade-off between detection accuracy and inference latency, where bubble sizes correspond to peak GPU memory consumption for each samples (average token length$\geq 100$). We observe that much external guardrails such as PIBert and DeBERTa achieve minimal latency ($<10$ ms) and memory usage ($<0.1$ GB). However, their low accuracy ($\approx 60\%$) renders them insufficient for defending against sophisticated injection attacks. However, LLM-based methods like LLM Naive achieve moderate accuracy but incur prohibitive latency overheads ($>1000$ ms). \aegis{} occupies the optimal position on the Pareto frontier among baselines. While it introduces a marginal latency increase ($\approx 120$ ms) compared to static classifiers, it achieves near-perfect detection accuracy ($>98\%$). Furthermore, \aegis{} remains computationally efficient compared to the substantial resource demands of distinct LLM-based defenses, demonstrating a favorable balance between robust safety and deployment feasibility.

\subsection{Ablation Studies}
\label{sec:ablation}

\begin{figure}[h]
    \centering
    \includegraphics[width=1\linewidth]{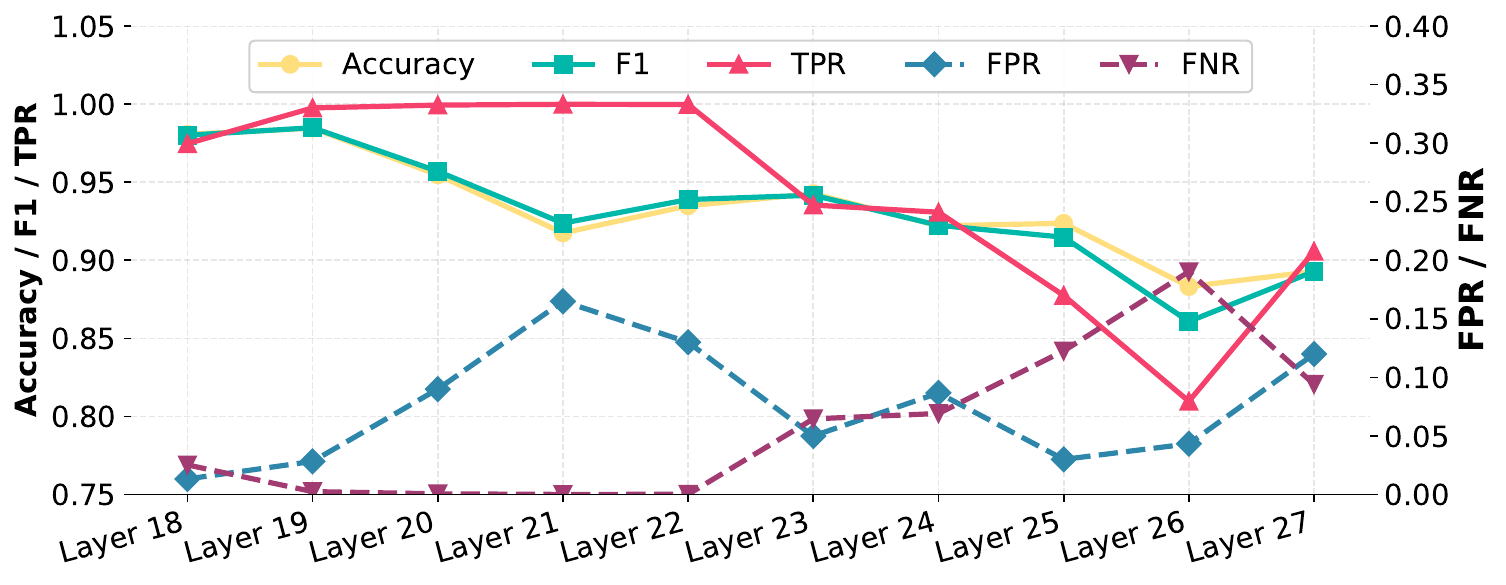}
    \vspace{-15pt}
    \caption{Detection performance of different layers.}
    \label{fig:layer_perforamce}
\end{figure}

\textbf{Performance Across Layers.} A core premise of our work is that instruction and knowledge representations exhibit distinct spectral signatures at specific depths of the network. Results from Qwen3-4B in Figure~\ref{fig:layer_perforamce} illustrate detection performance across layers 18 to 27. We observe that detection accuracy peaks and FPR is minimized in the middle layers, where the semantic separation between instructions and knowledge is most pronounced. While in deeper layers, the FNR increases significantly. This suggests that as the model approaches the output stage, the instruction manifold collapses into the token generation subspace, making the malicious intent harder to distinguish from benign generation. This finding justifies our strategy of targeting intermediate layers for projector extraction.

\begin{figure}[h]
    \centering
    \includegraphics[width=1\linewidth]{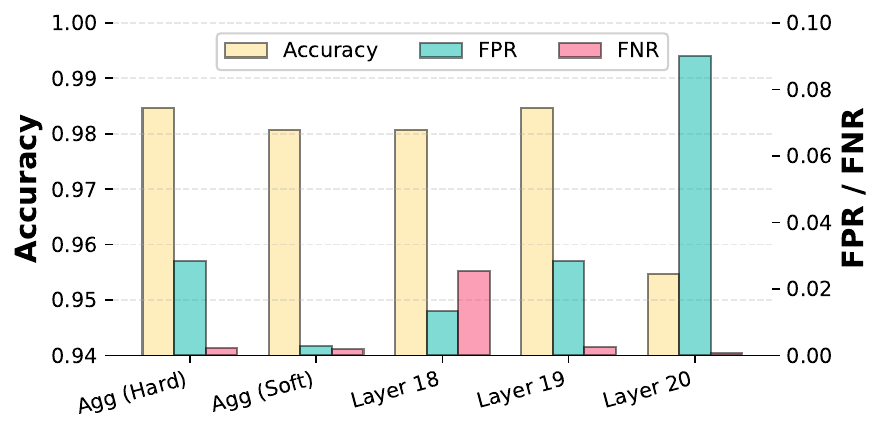}
    \vspace{-15pt}
    \caption{Ablation study of detection settings.}
    \label{fig:ablation_mode}
\end{figure}

\textbf{Impact of Multi-Layer Consensus.}
We evaluate the efficacy of the Unified Multi-Layer Consensus mechanism by comparing it against single-layer detectors. As shown in Figure~\ref{fig:ablation_mode}, result from Qwen3-4B, while individual layers (e.g., Layer 19) can achieve high accuracy, they often exhibit instability in either FPR or FNR. The ensemble approaches, both \textit{Agg (Hard)} and \textit{Agg (Soft)}, effectively mitigate this variance.  Specifically, the ensemble mechanisms maintain high accuracy ($>97\%$) while stabilizing the FPR below individual layer fluctuations. This confirms that aggregating topologically distinct signals filters out layer-specific noise, creating a robust detection boundary.

\begin{figure}[h]
    \centering
    \includegraphics[width=1\linewidth]{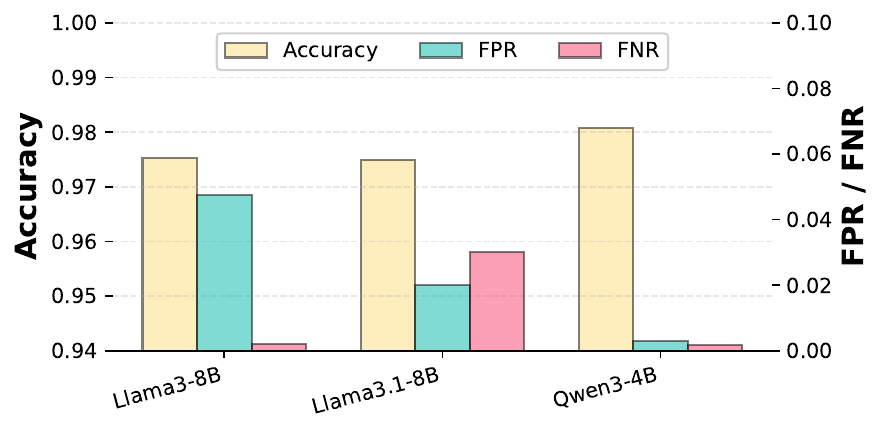}
    \vspace{-15pt}
    \caption{Detection performance of different models.}
    \label{fig:different_models}
\end{figure}

\textbf{Cross-Model Generalization.} To verify that \aegis{} captures fundamental properties of LLMs rather than architecture-specific artifacts, we evaluate its performance on other LLM families. Figure~\ref{fig:different_models} demonstrates the detection metrics on Llama-3-8B, Llama-3.1-8B, and Qwen3-4B. \aegis{} achieves consistent high performance (Accuracy $>97\%$, FPR $<5\%$) across all tested models. This universality suggests that the spectral asymmetry between instruction and knowledge is a generic structural invariant introduced by the alignment process, rendering \aegis{} applicable to a wide range of LLMs (More details in Table~\ref{tab:other_models}).

\begin{figure*}[t]
    \centering
    \includegraphics[width=0.9\linewidth]{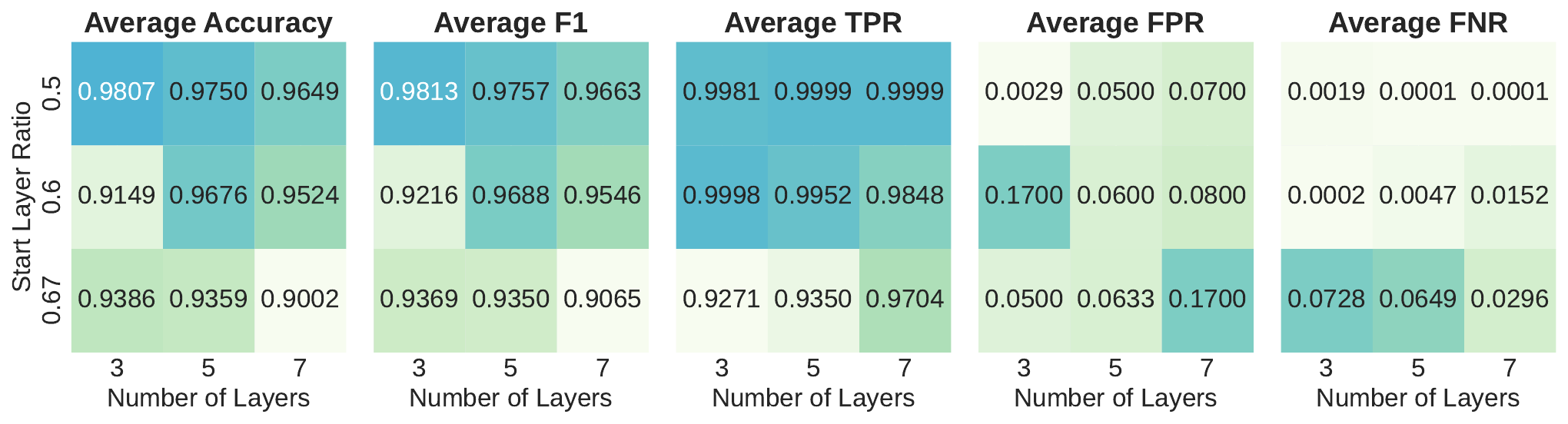}
    \vspace{-10pt}
    \caption{Detection performance (Qwen3-4B) with different numbers and positions of aggregated layers.}
    \label{fig:layer_heatmap}
\end{figure*}

\textbf{Sensitivity to Aggregation Hyperparameters.} Figure~\ref{fig:layer_heatmap} presents a heatmap analysis of the ensemble performance under varying start layer ratios (position in the network) and the number of voting layers ($|\mathcal{L}_{vote}|$). Starting aggregation at the midpoint (0.5) of the network yields the highest Average Accuracy and TPR. Shifting the window to later layers (Start Ratio 0.67) leads to a degradation in performance, characterized by increased FNR (lighter regions in the bottom-right heatmap). Furthermore, increasing the number of layers from 3 to 7 improves stability up to a saturation point.

\begin{figure}[h]
    \centering
    \includegraphics[width=1\linewidth]{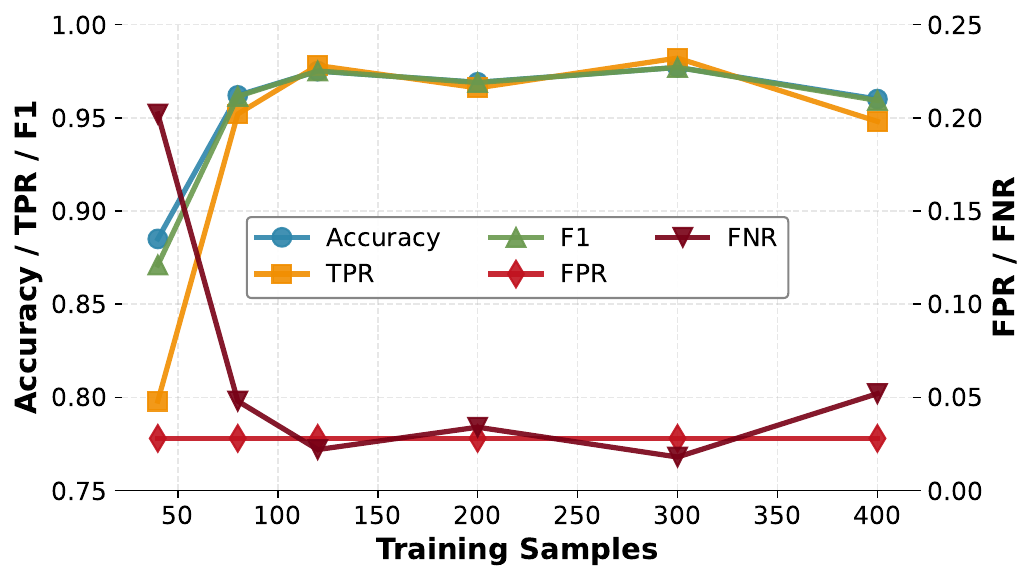}
    \vspace{-10pt}
    \caption{Detection performance (Llama3.1-8B) with different scales of training samples.}
    \label{fig:training_ratio}
\end{figure}

\textbf{Data Efficiency.} A practical defense must not require extensive datasets for calibration. We investigate the data efficiency of \aegis{} by varying the number of overall training samples (benign/malicious pairs) used to estimate the instruction-sensitive projector. Figure~\ref{fig:training_ratio} and Figure~\ref{fig:roc_curve} illustrate the performance and ROC curves, respectively, as the training sample size increases from 40 (20+20) to 400 (200+200). Remarkably, \aegis{} converges to optimal performance with as few as 100 samples. The accuracy and F1 scores plateau quickly, and the ROC curves in Figure~\ref{fig:roc_curve} show that the area under the curve (AUC) remains high even in low-data regimes. This high data efficiency proves that the instruction manifold is low-dimensional and stable, requiring minimal samples to characterize its principal components.

\begin{figure}[t]
    \centering
    \includegraphics[width=1\linewidth]{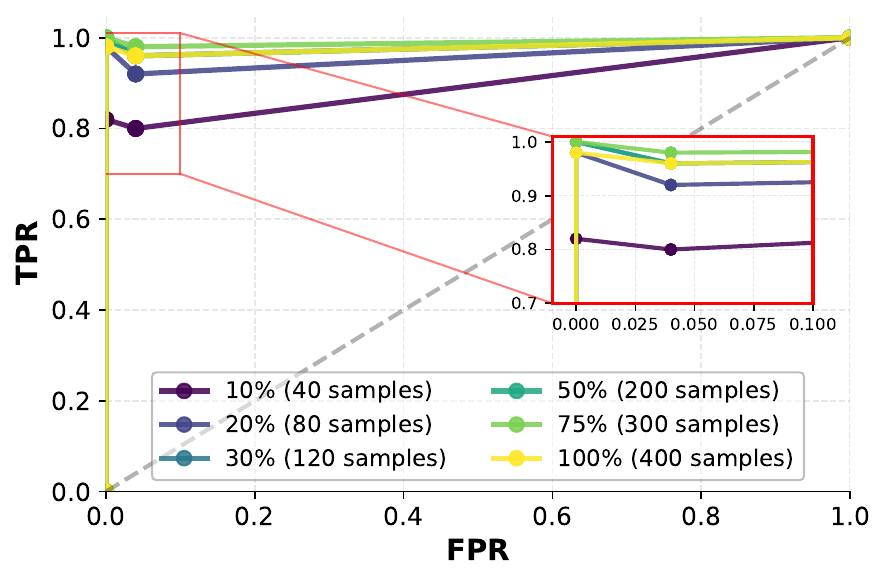}
    \vspace{-15pt}
    \caption{ROC curves (Llama3.1-8B) with different scales of training samples.}
    \label{fig:roc_curve}
    \vspace{-10pt}
\end{figure}

\section{Conclusion}
\label{sec:conclusion}
In this work, we analyze the conflation of instructions and data renders LLMs susceptible to PI. Our theoretical analysis identified the spectral asymmetry between the compact instruction manifold and the knowledge manifold as a critical leverage point. Based on this, we proposed \aegis{}, enabling the extraction of instruction-sensitive projectors. Furthermore, we formalized the robustness of our approach through a Unified Multi-Layer Consensus mechanism with layer-wise calibration and ensemble voting. Empirical evaluations across diverse benchmarks demonstrate that \aegis{} effectively detects indirect injections, achieving SOTA detection rates with negligible impact on benign utility.

\section*{Limitations}
While \aegis{} offers robust defense guarantees, we acknowledge several limitations that define the scope of its applicability: (1) \aegis{} requires access to the internal hidden states (activations) of the protected model, limiting its direct applicability to closed-source commercial models; (2) the error bound of \aegis{} is based on the assumption of layer independence, which might not always hold~\cite{abs-2502-02790}, leading to Equation~\ref{eq:error_bound} be an optimistic upper bound; (3) \aegis{} introduces a non-zero computational overhead during the detection, which may impact latency in ultra-high-throughput scenarios. These limitations can also be found in other IPI studies~\cite{liu2025datasentinel,jia2026promptlocate,chen2025struq,geng2025pisanitizer,piet2024jatmo,chen-etal-2025-defense}, regardless of their strategies.

\section*{Ethical Considerations}
We strictly adhere to the ACM Code of Ethics and the ACL Code of Conduct throughout the research. The primary objective of this study is to enhance the security and robustness of LLMs by proposing \aegis{} to mitigate the risks of IPI in real-world applications. Our experiments rely exclusively on established public benchmarks. All data used in this study is anonymized and contains no personally identifiable information. We do not introduce or release any new datasets containing harmful, toxic, or offensive content beyond what is already available in the public domain for research purposes. Our commitment to transparency is reflected in making \aegis{} fully open-source upon publication, fostering collaboration, and allowing others to verify, replicate, and build upon our work for the betterment of the field. 

While our work focuses on enhancing defensive mechanisms against prompt injection attacks, we acknowledge the potential for dual use in security research. We encourage the ethical and responsible use of \aegis{} to improve LLM security and not for malicious purposes. 

\section*{LLM Usage Considerations}
We utilized Large Language Models (specifically Gemini3-Pro and GPT-5) to assist with checking grammatical errors, refining sentence structure, and polishing the clarity of the text. All scientific claims, experimental designs, theoretical derivations, and the final decision on the content remain the sole responsibility of the human authors.


\bibliography{acl_latex.bbl}

\appendix
\section{Appendix}
\label{sec:appendix}

\subsection{Derivation of the Ensemble Error}
\label{app:proof_hoeffding}
Here, we provide the derivation of the exponential error bound for the \aegis{} ensemble under \textbf{Hard Voting} regime, as presented in Eq.~\ref{eq:error_bound}. For a benign input $x$, let $X_l$ be the binary random variable representing the decision of the $l$-th layer's detector:
\begin{equation}
    X_l = \phi_{hard}(S_l(x)) = \mathbb{I}(S_l(x) > b_l)
\end{equation}
where $X_l=1$ indicates a False Positive (FP) at layer $l$, and $X_l=0$ indicates a correct rejection.

\paragraph{Assumption 1 (Bounded Single-Layer Error).}
Due to the distribution-aware calibration described in Sec.~\ref{sec:methodology}, the FPR of each individual layer is strictly bounded by $\beta$:
\begin{equation}
    P(X_l = 1) = \mathbb{E}[X_l] \le \beta
\end{equation}

\paragraph{Assumption 2 (Approximate Independence).}
We model the layer-wise decisions $\{X_l\}_{l=1}^L$ as independent Bernoulli random variables. While deeper layers in LLMs exhibit functional dependencies, non-adjacent layers often capture distinct semantic features~\cite{zhang2024investigating}. This independence assumption provides a theoretical upper bound for the worst-case analysis.

We aim to bound the Ensemble FPR, $Q_{ens}$, defined as the probability that the sum of votes exceeds the consensus quorum $\Gamma$:
\begin{equation}
    Q_{ens} = P\left( \sum_{l=1}^{L} X_l \ge \Gamma \right)
\end{equation}

Let $S_L = \sum_{l=1}^{L} X_l$ be the total number of positive votes. Under the independence assumption, $S_L$ follows a sum of independent bounded random variables. The number of false positives is:
\begin{equation}
    \mathbb{E}[S_L] = \sum_{l=1}^{L} \mathbb{E}[X_l] \le L\beta
\end{equation}
We apply \textbf{Hoeffding's Inequality}~\cite{hoeffding1963probability}, which provides an upper bound on the probability that the sum of bounded independent random variables deviates from its expected value by more than a certain amount $t$.

\begin{lemma}[Hoeffding's Inequality]
Let $X_1, \dots, X_n$ be independent random variables such that $a_i \le X_i \le b_i$. Let $S_n = \sum X_i$. Then for any $t > 0$:
\begin{equation}
    P(S_n - \mathbb{E}[S_n] \ge t) \le \exp\left( -\frac{2t^2}{\sum_{i=1}^n (b_i - a_i)^2} \right)
\end{equation}
\end{lemma}

In our specific case: (1) The variables are binary, so $a_i=0, b_i=1$, and $(b_i-a_i)^2 = 1$; (2) The denominator becomes $\sum_{i=1}^L 1 = L$; (3) We are interested in the probability of the sum exceeding $\Gamma$. We rewrite the event $S_L \ge \Gamma$ in terms of deviation from the expectation $L\beta$.

Let $t = \Gamma - L\beta$. We assume the voting threshold is set higher than the expected single-layer error rate (i.e., $\Gamma > L\beta$), ensuring $t > 0$. We substitute these into Hoeffding's bound:

\begin{equation}
\begin{aligned}
    Q_{ens} &= P(S_L \ge \Gamma) \\
            &= P(S_L - L\beta \ge \Gamma - L\beta) \\
            &\le P(S_L - \mathbb{E}[S_L] \ge \Gamma - L\beta)
\end{aligned}
\end{equation}

Applying the Lemma with $t = \Gamma - L\beta$:

\begin{equation}
    Q_{ens} \le \exp\left( -\frac{2(\Gamma - L\beta)^2}{L} \right)
\end{equation}

To provide a more intuitive form related to the \textit{voting ratio} versus the \textit{error rate}, we factor out $L$ from the numerator term $(\Gamma - L\beta)^2$:

\begin{equation}
\begin{aligned}
    (\Gamma - L\beta)^2 &= \left( L \left( \frac{\Gamma}{L} - \beta \right) \right)^2 \\
    &= L^2 \left( \frac{\Gamma}{L} - \beta \right)^2
\end{aligned}
\end{equation}

Substituting this back into the exponent:

\begin{equation}
    \begin{aligned}
    Q_{ens} &\le \exp\left( -\frac{2 \cdot L^2 \left( \frac{\Gamma}{L} - \beta \right)^2}{L} \right) \\
            &= \exp\left( -2L \left( \frac{\Gamma}{L} - \beta \right)^2 \right)
    \end{aligned}
\end{equation}

Noting that $L = |\mathcal{L}_{vote}|$, we arrive at the final bound presented in the main text:

\begin{equation}
    Q_{ens} \le \exp\left( -2 |\mathcal{L}_{vote}| \left( \frac{\Gamma}{|\mathcal{L}_{vote}|} - \beta \right)^2 \right)
\end{equation}

\hfill $\blacksquare$

\textbf{Remark on Convergence.}
This result demonstrates that the ensemble error rate decays exponentially with the number of layers $|\mathcal{L}_{vote}|$. The decay rate is controlled by the square of the \textbf{Safety Margin} $(\frac{\Gamma}{L} - \beta)$. As long as the consensus requirement (e.g., $\Gamma/L = 0.5$) is distinct from the single-layer noise floor ($\beta \approx 0.1$), the probability of a false alarm vanishes as more layers are added.

\subsection{Stability of Soft Voting}
\label{app:stable_soft}
In the \textbf{Soft Voting} regime ($\phi_{soft}$), \aegis{} operates as a signal accumulator. Consider a stealthy attack that induces a weak activation $\mu_{adv}$ at each layer, submerged in background noise $\sigma_{noise}$. While the single-layer Signal-to-Noise Ratio (SNR) is $\mu_{adv}/\sigma_{noise}$, the ensemble sums the consistent semantic signal linearly ($\propto |\mathcal{L}|$) while the independent noise components sum orthogonally ($\propto \sqrt{|\mathcal{L}|}$). The ensemble SNR scales by $\sqrt{|\mathcal{L}|}$, empowering \aegis{} to detect micro-perturbations that are statistically indistinguishable from noise in any single projection but conspicuous in the aggregate.

\subsection{Poof of Ensemble Robustness Radius}
\begin{proposition}[Ensemble Latent Robustness Radius]
    Let $\mathcal{R}(x) = \{R_1(x), R_2(x), \dots, R_{|\mathcal{L}|}(x)\}$ be the set of exact latent radii for all monitored layers calculated via Eq.~(\ref{eq:exact_radius}). Let $R_{(\Gamma)}$ denote the $\Gamma$-th smallest value in $\mathcal{R}(x)$ (i.e., the $\Gamma$-th order statistic).
    The robustness radius of the \aegis{} ensemble, $R_{ens}(x)$, is determined by:
    \begin{equation}
        R_{ens}(x) = R_{(\Gamma)}
    \end{equation}
    Any latent perturbation $\delta$ with $\|\delta\|_2 < R_{(\Gamma)}$ is mathematically guaranteed to preserve the safe classification of the ensemble.
\end{proposition}

\begin{proof}
    Consider an arbitrary perturbation $\delta$ with norm $\|\delta\|_2 < R_{(\Gamma)}$.
    By the definition of order statistics, there exist at most $\Gamma - 1$ layers where the layer-specific safety margin $R_l(x)$ is smaller than or equal to $\|\delta\|_2$.
    For all other layers (at least $|\mathcal{L}| - \Gamma + 1$ layers), the perturbation is insufficient to cross the decision boundary:
\begin{equation}
\begin{aligned}
    \langle h^{(l)} + \delta, w_l^* \rangle &\le \langle h^{(l)}, w_l^* \rangle + \|\delta\|_2 \\
    &< S_l(x) + R_l(x) = b_l
\end{aligned}
\end{equation}
    Consequently, the number of positive votes is strictly less than the quorum $\Gamma$. The ensemble decision function $\mathbb{I}(\sum v_l \ge \Gamma)$ evaluates to 0, ensuring the input remains classified as safe.
\end{proof}

\begin{figure*}
    \centering
    \includegraphics[width=0.78\linewidth]{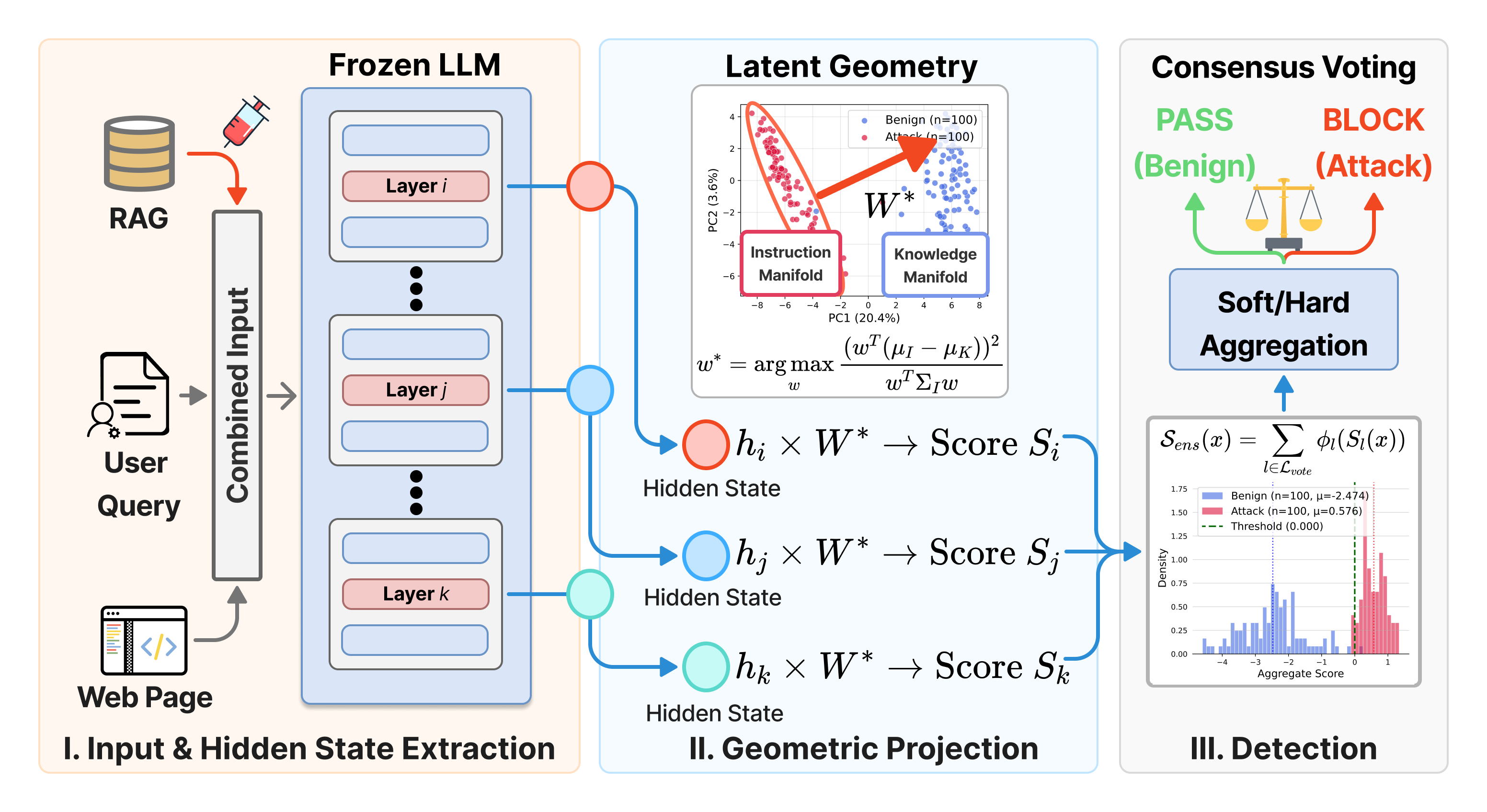}
    \caption{Conceptual illustration of the \aegis{}.}
    \label{fig:framework}
\end{figure*}

\subsection{Details of \aegis{}}
\label{app:aegis_detail}
We conduct our experiments using PyTorch 2.4.0. All experiments are performed on a single NVIDIA A100-8G GPU.
For the default setting of Qwen3-4B in the main experiments, we select the start position as \textbf{0.5} (representing the relative network depth, e.g., layer 16 for a 32-layer model) and employ a voting window of $|\mathcal{L}_{vote}| = 3$ layers. For Llama3-8B and Llama3.1-8B, we select the start position as \textbf{0.5} and employ a voting window of $|\mathcal{L}_{vote}| = 7$ layers. By default, we utilize the \textbf{Soft Voting} mechanism as the default aggregation strategy. The covariance matrices $\Sigma_I$ and projectors are estimated using a computationally efficient calibration set consisting of only 200 benign samples and 200 malicious samples randomly sampled from the training corpus. The threshold to limit the impact on utility strictly bration was selected by 0.05 FPR on the hold-out benign set (training set). The selected threshold was applied to all test datasets and attacks all the time. Note that there are no sample overlaps between the training corpus and the test set. The evaluations for both \aegis{} and all baselines were run 5 times with different random seeds, and we report the mean value of the final results.

\subsection{Implementation of Baselines}
\label{app:baseline_detail}
We compare \aegis{} against a comprehensive suite of industrial guardrails and research-based detection methods. All baselines are implemented following their official repositories or standard configurations described in their respective papers. For all methods, thresholds are tuned on the same validation benign set to achieve target FPR (1\% unless otherwise stated). Test FPR/FNR are reported on a disjoint test set.

\noindent\textbf{External Classifiers.}
\begin{itemize}
    \item \textbf{PIBert}~\cite{pibert}: We utilize the \texttt{fmops/distilbert-prompt-injection} checkpoint, a DistilBERT model fine-tuned specifically for detecting prompt injection attacks.
    \item \textbf{DeBERTa}~\cite{deberta-pi}: We employ the deberta-v3-base-prompt-injection model, which is optimized for identifying heuristic injection patterns.
    \item \textbf{PromptGuard (PrGuard)}~\cite{PromptGuard}: We use \texttt{meta-llama/Prompt-Guard-86M}, a lightweight classifier designed to protect LLM-based applications from malicious inputs.
\end{itemize}
For all external classifiers, we apply their original (default) classification threshold. Note that this is also the optimal setting given by the official recommendation.

\noindent\textbf{Perplexity-based Detectors.} Following Alon et al.~\cite{alon2023detecting}, we implement two variants based on the hypothesis that adversarial inputs exhibit anomalous perplexity:
\begin{itemize}
    \item \textbf{PPL\textsuperscript{1}}: Uses GPT-2-Small as the reference model to calculate windowed perplexity.
    \item \textbf{PPL\textsuperscript{2}}: Uses Vicuna-7B as a more capable reference model.
\end{itemize}

\noindent\textbf{Internal \& LLM-based Defenses.}
\begin{itemize}
    \item \textbf{InjecGuard}~\cite{li2024injecguard}: We leverage the official opensourced checkpoint and code for implementation.
    \item \textbf{Attention Tracker (AttTracker)}~\cite{hung2025attention}: We implement the attention monitoring mechanism that flags inputs when the attention weights assigned to the external context exceed a specified threshold, indicating a ``hijacking'' of the model's focus.
    \item \textbf{LLM Naive}~\cite{llm-naive}: We employ a ``defense-via-prompting'' approach where a separate LLM (Vicuna-7B-V1.3) serves as a judge. The model is prompted with a system instruction to classify the user input.
    \item \textbf{Known-Answer (KA)}~\cite{Known-Answer}: We verify the integrity of the system prompt by appending a secret ``canary'' token sequence. If the model's output fails to reproduce the canary correctly (indicating the instruction was overridden), the input is flagged as an attack.
    \item \textbf{PIShield}~\cite{zou2025pishield}: We implement this defense by extracting the internal representation of the final prompt token from a specific ``injection-critical layer'' (middle layer in this paper). A simple linear classifier (MLP) is trained on these latent features to distinguish between clean and contaminated prompts based on their intrinsic spectral differences. Note that for fair comparison, we do not include all the training data (20000 samples) and align with \aegis{} using 200 samples. This might be why PIShield starkly contradicts the original results. Note that we also have to admit that PIShield also exhibits good performance (close but lower than \aegis{}) when full datasets are leveraged for training, but this validate the data efficiency of \aegis{} over PIShield as well.
\end{itemize}

\begin{figure*}
    \centering
    \includegraphics[width=1\linewidth]{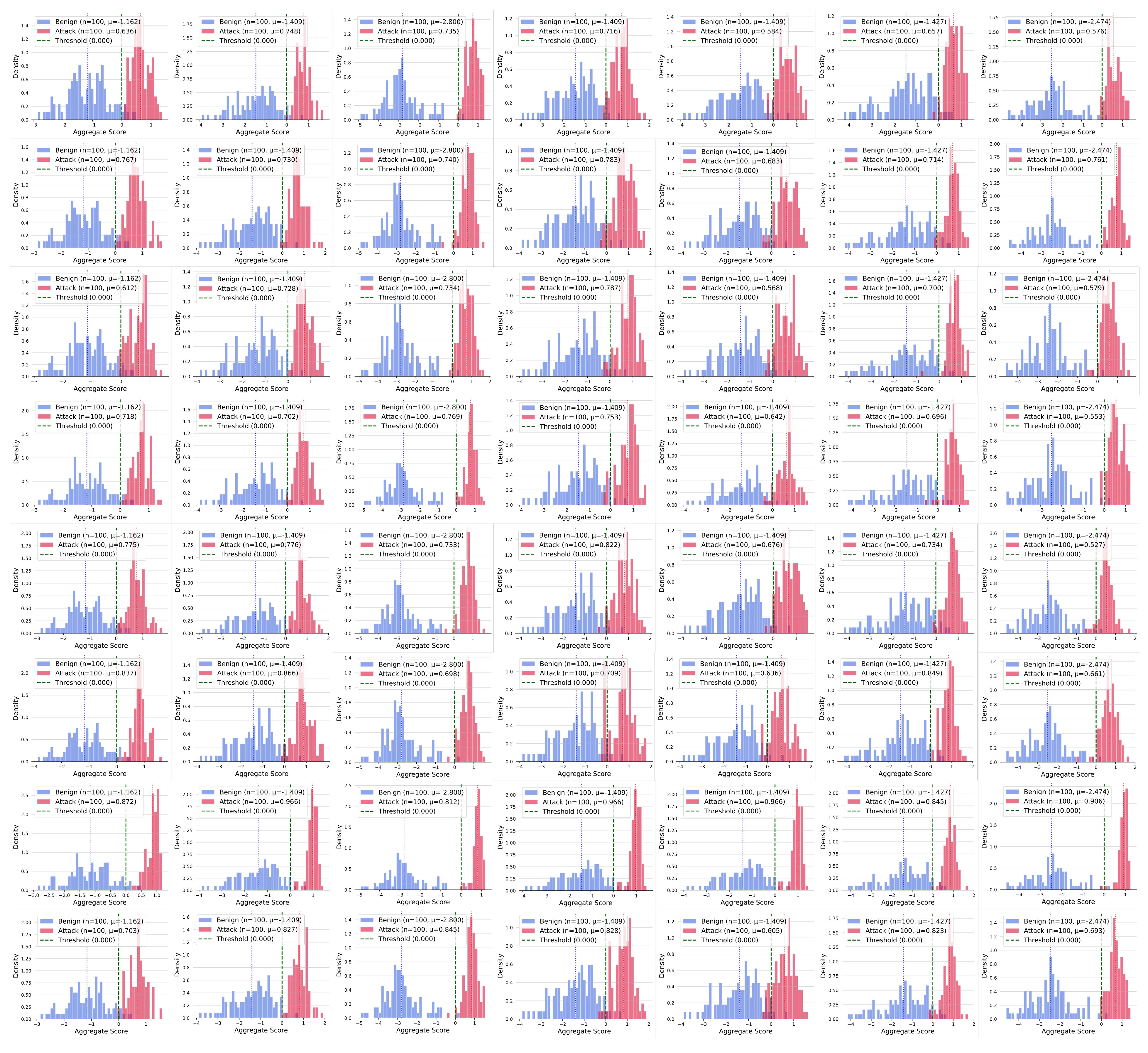}
    \caption{Distribution of the soft voting score by \aegis{}. From top to bottom, the IPI attacks are: \texttt{Combined Attack, Escape Character, Fake Completion, Context Ignoring, Naive Attack, NeuralExec, PLeak} and \texttt{Universal Injection}. From left to right, the corresponding datasets are: \texttt{gigaword, hsol, jfleg, mrpc, rte, sms-spam} and \texttt{sst2}}
    \label{fig:hist_distribution}
\end{figure*}

\begin{table*}[t]
\centering
\setlength{\tabcolsep}{0.8mm}
\renewcommand{\arraystretch}{1.2}
\aboverulesep=0ex
\belowrulesep=0.5ex
\resizebox{1.0\linewidth}{!}{
\begin{tabular}{cccccccccccccccccccc}
\toprule[2pt]
\multirow{3}{*}{\textbf{Attack} $\downarrow$} & \textbf{Model} $\rightarrow$ & \multicolumn{9}{c}{\textbf{Llama3-8B}} & \multicolumn{9}{c}{\textbf{Llama3.1-8B}} \\ \cmidrule(lr){3-11} \cmidrule(lr){12-20} 
 & \textbf{Start Ratio} $\rightarrow$ & \multicolumn{3}{c}{0.5} & \multicolumn{3}{c}{0.67} & \multicolumn{3}{c}{0.6} & \multicolumn{3}{c}{0.5} & \multicolumn{3}{c}{0.67} & \multicolumn{3}{c}{0.6} \\ \cmidrule(lr){3-5}\cmidrule(lr){6-8}\cmidrule(lr){9-11}\cmidrule(lr){12-14}\cmidrule(lr){15-17}\cmidrule(lr){18-20}
 & \textbf{Num Layers} $\rightarrow$ & 3 & 5 & 7 & 3 & 5 & 7 & 3 & 5 & 7 & 3 & 5 & 7 & 3 & 5 & 7 & 3 & 5 & 7 \\ \midrule
\multirow{3}{*}{\makecell{\textbf{Combined}\\ \textbf{Attack}}} & ACC $\uparrow$ & 0.9712 & 0.9736 & 0.9761 & 0.9713 & 0.9640 & 0.9739 & 0.9705 & 0.9647 & 0.9604 & 0.8823 & 0.9689 & 0.9704 & 0.9564 & 0.9521 & 0.9549 & 0.9727 & 0.9684 & 0.9609 \\
 & FPR $\downarrow$ & 0.0575 & 0.0525 & 0.0475 & 0.0540 & 0.0700 & 0.0500 & 0.0550 & 0.0680 & 0.0780 & 0.0120 & 0.0180 & 0.0200 & 0.0260 & 0.0300 & 0.0340 & 0.0240 & 0.0220 & 0.0280 \\
 & FNR $\downarrow$ & 0.0000 & 0.0004 & 0.0004 & 0.0034 & 0.0020 & 0.0021 & 0.0039 & 0.0026 & 0.0011 & 0.2234 & 0.0443 & 0.0391 & 0.0611 & 0.0657 & 0.0563 & 0.0306 & 0.0411 & 0.0503 \\ \hline
\multirow{3}{*}{\makecell{\textbf{Escape}\\ \textbf{Character}}} & ACC $\uparrow$ & 0.9707 & 0.9727 & 0.9752 & 0.9696 & 0.9623 & 0.9725 & 0.9700 & 0.9639 & 0.9589 & 0.9473 & 0.9800 & 0.9746 & 0.9633 & 0.9601 & 0.9624 & 0.9724 & 0.9697 & 0.9659 \\
 & FPR $\downarrow$ & 0.0575 & 0.0525 & 0.0475 & 0.0540 & 0.0700 & 0.0500 & 0.0550 & 0.0680 & 0.0780 & 0.0120 & 0.0180 & 0.0200 & 0.0260 & 0.0300 & 0.0340 & 0.0240 & 0.0220 & 0.0280 \\
 & FNR $\downarrow$ & 0.0011 & 0.0021 & 0.0021 & 0.0069 & 0.0054 & 0.0050 & 0.0050 & 0.0043 & 0.0043 & 0.0934 & 0.0220 & 0.0309 & 0.0474 & 0.0497 & 0.0411 & 0.0311 & 0.0386 & 0.0403 \\ \hline
\multirow{3}{*}{\makecell{\textbf{Fake}\\ \textbf{Completion}}} & ACC $\uparrow$ & 0.9709 & 0.9732 & 0.9759 & 0.9701 & 0.9641 & 0.9739 & 0.9713 & 0.9647 & 0.9606 & 0.9196 & 0.9714 & 0.9706 & 0.9583 & 0.9544 & 0.9576 & 0.9710 & 0.9656 & 0.9606 \\
 & FPR $\downarrow$ & 0.0575 & 0.0525 & 0.0475 & 0.0540 & 0.0700 & 0.0500 & 0.0550 & 0.0680 & 0.0780 & 0.0120 & 0.0180 & 0.0200 & 0.0260 & 0.0300 & 0.0340 & 0.0240 & 0.0220 & 0.0280 \\
 & FNR $\downarrow$ & 0.0007 & 0.0011 & 0.0007 & 0.0057 & 0.0017 & 0.0021 & 0.0025 & 0.0026 & 0.0009 & 0.1489 & 0.0391 & 0.0389 & 0.0574 & 0.0611 & 0.0509 & 0.0340 & 0.0469 & 0.0509 \\ \hline
\multirow{3}{*}{\makecell{\textbf{Context}\\ \textbf{Ignore}}} & ACC $\uparrow$ & 0.9707 & 0.9727 & 0.9752 & 0.9709 & 0.9641 & 0.9734 & 0.9711 & 0.9647 & 0.9600 & 0.9229 & 0.9753 & 0.9707 & 0.9596 & 0.9571 & 0.9604 & 0.9723 & 0.9683 & 0.9621 \\
 & FPR $\downarrow$ & 0.0575 & 0.0525 & 0.0475 & 0.0540 & 0.0700 & 0.0500 & 0.0550 & 0.0680 & 0.0780 & 0.0120 & 0.0180 & 0.0200 & 0.0260 & 0.0300 & 0.0340 & 0.0240 & 0.0220 & 0.0280 \\
 & FNR $\downarrow$ & 0.0011 & 0.0021 & 0.0021 & 0.0043 & 0.0017 & 0.0032 & 0.0029 & 0.0026 & 0.0020 & 0.1423 & 0.0314 & 0.0386 & 0.0549 & 0.0557 & 0.0451 & 0.0314 & 0.0414 & 0.0477 \\ \hline
\multirow{3}{*}{\makecell{\textbf{Naive}\\ \textbf{Attack}}} & ACC $\uparrow$ & 0.9705 & 0.9727 & 0.9752 & 0.9691 & 0.9624 & 0.9723 & 0.9704 & 0.9636 & 0.9590 & 0.9487 & 0.9801 & 0.9754 & 0.9657 & 0.9631 & 0.9644 & 0.9741 & 0.9726 & 0.9680 \\
 & FPR $\downarrow$ & 0.0575 & 0.0525 & 0.0475 & 0.0540 & 0.0700 & 0.0500 & 0.0550 & 0.0680 & 0.0780 & 0.0120 & 0.0180 & 0.0200 & 0.0260 & 0.0300 & 0.0340 & 0.0240 & 0.0220 & 0.0280 \\
 & FNR $\downarrow$ & 0.0014 & 0.0021 & 0.0021 & 0.0077 & 0.0051 & 0.0054 & 0.0043 & 0.0049 & 0.0040 & 0.0906 & 0.0217 & 0.0291 & 0.0426 & 0.0437 & 0.0371 & 0.0277 & 0.0329 & 0.0360 \\ \hline
\multirow{3}{*}{\textbf{NeuralExec}} & ACC $\uparrow$ & 0.9705 & 0.9721 & 0.9746 & 0.9676 & 0.9606 & 0.9696 & 0.9679 & 0.9619 & 0.9584 & 0.9047 & 0.9731 & 0.9731 & 0.9653 & 0.9609 & 0.9629 & 0.9721 & 0.9731 & 0.9646 \\
 & FPR $\downarrow$ & 0.0575 & 0.0525 & 0.0475 & 0.0540 & 0.0700 & 0.0500 & 0.0550 & 0.0680 & 0.0780 & 0.0120 & 0.0180 & 0.0200 & 0.0260 & 0.0300 & 0.0340 & 0.0240 & 0.0220 & 0.0280 \\
 & FNR $\downarrow$ & 0.0014 & 0.0032 & 0.0032 & 0.0109 & 0.0089 & 0.0107 & 0.0093 & 0.0083 & 0.0051 & 0.1786 & 0.0357 & 0.0337 & 0.0434 & 0.0483 & 0.0403 & 0.0317 & 0.0317 & 0.0429 \\ \hline
\multirow{3}{*}{\textbf{PLeak}} & ACC $\uparrow$ & 0.9700 & 0.9725 & 0.9737 & 0.9610 & 0.9550 & 0.9638 & 0.9700 & 0.9560 & 0.9520 & 0.9550 & 0.9890 & 0.9880 & 0.9850 & 0.9840 & 0.9820 & 0.9860 & 0.9870 & 0.9840 \\
 & FPR $\downarrow$ & 0.0575 & 0.0525 & 0.0475 & 0.0540 & 0.0700 & 0.0500 & 0.0550 & 0.0680 & 0.0780 & 0.0120 & 0.0180 & 0.0200 & 0.0260 & 0.0300 & 0.0340 & 0.0240 & 0.0220 & 0.0280 \\
 & FNR $\downarrow$ & 0.0025 & 0.0025 & 0.0050 & 0.0240 & 0.0200 & 0.0225 & 0.0050 & 0.0200 & 0.0180 & 0.0780 & 0.0040 & 0.0040 & 0.0040 & 0.0020 & 0.0020 & 0.0040 & 0.0040 & 0.0040 \\ \hline
\multirow{3}{*}{\makecell{\textbf{Universal}\\\textbf{Injection}}} & ACC $\uparrow$ & 0.9712 & 0.9734 & 0.9761 & 0.9666 & 0.9601 & 0.9707 & 0.9707 & 0.9627 & 0.9577 & 0.9119 & 0.9793 & 0.9769 & 0.9681 & 0.9599 & 0.9596 & 0.9747 & 0.9734 & 0.9659 \\
 & FPR $\downarrow$ & 0.0575 & 0.0525 & 0.0475 & 0.0540 & 0.0700 & 0.0500 & 0.0550 & 0.0680 & 0.0780 & 0.0120 & 0.0180 & 0.0200 & 0.0260 & 0.0300 & 0.0340 & 0.0240 & 0.0220 & 0.0280 \\
 & FNR $\downarrow$ & 0.0000 & 0.0007 & 0.0004 & 0.0129 & 0.0097 & 0.0086 & 0.0036 & 0.0066 & 0.0066 & 0.1643 & 0.0234 & 0.0263 & 0.0377 & 0.0503 & 0.0469 & 0.0266 & 0.0311 & 0.0403 \\ \bottomrule[1pt]
\end{tabular}}
\caption{Detection performance on Llama3-8B and Llama3.1-8B regarding different number of layers and position.}
\vspace{-8pt}
\label{tab:other_models}
\end{table*}

\end{document}